\documentclass[preprint, 3p, dvipdfmx, sort&compress]{elsarticle} 

\usepackage{amsmath, amssymb, amsthm}
\usepackage{bm, physics}
\usepackage{subcaption}
\usepackage{booktabs, diagbox}
\usepackage{comment}

\newtheorem{define}{Definition}
\newtheorem{theorem}{Theorem}

\newcommand{\relmiddle}[1][|]{\mathrel{}\middle#1\mathrel{}}
\newcommand{\tlin}{\text{lin}}
\newcommand{\tconv}{\text{conv}}
\newcommand{\tconc}{\text{conc}}

\title{Asymptotically stable matchings and evolutionary dynamics of preference revelation games in marriage problems}

\author[1]{Hidemasa Ishii\corref{cor1}
}
\ead{hidemasaishii1997@g.ecc.u-tokyo.ac.jp}

\author[2]{Nariaki Nishino
}

\cortext[cor1]{Corresponding author.}
\address[1]{Graduate School of Frontier Science, The University of Tokyo, Kashiwa-shi, Chiba 277-8561, Japan}
\address[2]{School of Engineering, The University of Tokyo, Bunkyo-ku, Tokyo 113-8656, Japan}

\begin{document}
    \begin{abstract}
        The literature on centralized matching markets often assumes
        that a true preference of each player is known to herself and fixed,
        but empirical evidence casts doubt on its plausibility.
        To circumvent the problem, we consider evolutionary dynamics of preference revelation games 
        in marriage problems.
        We formulate the asymptotic stability of a matching, indicating the dynamical robustness 
        against sufficiently small changes in players' preference reporting strategies,
        and show that asymptotically stable matchings are stable when they exist.
        The simulation results of replicator dynamics are presented to demonstrate the asymptotic stability.
        We contribute a practical insight for market designers that a stable matching may be realized by introducing a learning period in which participants find appropriate reporting strategies through trial and error.
        We also open doors to a novel area of research by demonstrating ways to employ evolutionary game theory in studies on centralized markets.
    \end{abstract}

    \maketitle


\section{Introduction} \label{sec:intro}
The marriage problem is a simple and basic model of two-sided matching markets 
which involves one-to-one matchings.
Suppose that there are two sets of players $M$ and $W$, often referred to as men and women respectively.
Each player has a preference relation defined over the other set and remaining single.
A preference profile $P$ is a vector which consists of the preference relations of all players.
A matching is said to be stable, if no player is worse off than to remain unmatched, 
and if no pair of a man and a woman prefer each other to their current partners.
Gale and Shapley \cite{gale1962} formulated the marriage problem and 
showed that the set of stable matchings is always nonempty.
In the proof, they presented the so-called deferred acceptance (DA) algorithm, which always yields a stable matching.
They also showed the existence of the $M$-optimal stable matching: 
i.e. a stable matching that is preferred to all the other stable matchings by all men.
Similarly, there exists the $W$-optimal stable matching.

In practice, a matching has to be determined based on the reported preferences,
because a matchmaker can never know the true preferences of players.
A matching algorithm induces the preference revelation game, 
a noncooperative game where players submit their preferences strategically.
We refer to a preference profile submitted by the players as a report profile,
in order to distinguish it from the true preference profile.
A number of studies have 
investigated incentives of players to report their true preferences \cite{roth1982,dubins1981,gale1985a,immorlica2005,kojima2009} and 
characterized the set of matchings resulting from the strategic reporting behavior \cite{roth1984,ma1995,alcalde1996,sonmez1997}.
These studies considered centralized matching markets where participants submit their preferences to a clearinghouse based on which an algorithm generates the matching.
Another line of research concerns decentralized markets and has analyzed dynamic processes in which matchings are updated in a pairwise manner \cite{roth1990,jackson2002,newton2015,bilancini2020}.

The literature on centralized markets often assumes that a true preference of each player is known to herself and fixed throughout the matching process.
However several empirical works on school choice situations cast doubt on its plausibility.
Dwenger et al. \cite{dwenger2018} concluded that participants in the university admissions centralized market in Germany chose to randomize their decisions.
The preliminary results by Narita \cite{narita2018} and Grenet et al. \cite{grenet2021} both suggest that preferences of players changed in time and learning was the prominent driver of those changes.
The current paper explores the possibility of circumventing this problem by adopting an evolutionary uuapproach.
The standard setting of game theory is that the game is played exactly once by fully rational players.
The contrasting setting of evolutionary game theory \cite{weibull1997,newton2018} is that players are not necessarily rational and their behavior changes in time through such processes as trial-and-error and social learning.
Evolutionary dynamics of a game refers to such time evolution of the strategy profile.
Since players learn their own and others' preferences by experience, they need not be aware of their preferences when the process begins, unlike much of the literature on centralized markets.

In this paper, we theoretically investigate what kind of matchings will emerge through evolutionary dynamics of preference revelation games in marriage problems.
We consider a class of continuous-time deterministic evolutionary dynamics which contains the famous replicator dynamics.
We define the asymptotic stability of a matching, an analogue to the asymptotic stability of a fixed point in dynamical systems.
When a matching is asymptotically stable, it is likely to be obtained at the equilibria of evolutionary dynamics.
Furthermore, it turns out that an asymptotically stable matching may not exist, but when it does, it is also stable.
The simulation results of the replicator dynamics are presented to demonstrate the dynamical robustness of an asymptotically stable matching against perturbation in the strategy profile.
The asymptotic stability of a matching indicates the matching will be restored after another matching is realized due to sufficiently small changes in players' reporting strategies.
When a fixed point of a dynamical system is asymptotically stable, it intuitively means that sufficiently small perturbation results in a movement that goes back to the fixed point\footnote{%
Let $x^*$ be a fixed point of a system $\dot{x} = f(x)$.
$x^*$ is Lyapunov stable if, for each $\epsilon > 0$, there exists a $\delta > 0$ such that if $\norm{x(0) - x^*} < \delta$ then $\norm{x(t) - x^*} < \epsilon$ for all $t \geq 0$.
$x^*$ is attracting if there exists $\delta > 0$ such that if $\norm{x(0) - x^*} < \delta$ then $\lim_{t \to \infty} x(t) = x^*$.
$x^*$ is asymptotically stable if it is both Lyapunov stable and attracting \cite{strogatz2015}.
}.
Our asymptotic stability of a matching is similar to the asymptotic stability in dynamical systems, but differs in that a focal state and perturbation are not in the same space.
In matching problems, a focal state, a matching, is a point in the matching space, which is the set of all possible matchings.
By contrast, perturbation corresponds to a change in the reporting strategy of a player, and players' strategies are represented by a point in the strategy space.

This research contributes to matching theory in practical and theoretical manner.
Firstly, we formulated the asymptotic stability of a matching, which characterize the matchings that emerge through evolutionary dynamics.
An asymptotically stable matching is likely to be realized when players experiment or make mistakes occasionary.
Furthermore we found that the asymptotic stability of a matching implies its stability.
Thus the current paper suggests a novel practical strategy for market designers to implement stable matchings, 
which is to create a learning phase in which participants can try submitting various preferences and are informed of the would-be matchings many times over.
Players would find reporting strategies that appropriately represent their true preferences during this learning phase.
Accordingly an asymptotically stable and hence stable matching is realized.
We emphasize that our argument does not neccesitate players' initial awareness of their true preferences, unlike the literature.
Secondly, we demonstrate how to employ evolutionary game theory in studying centralized matching markets, thereby opening doors to a novel area of research.
Many directions can be taken: 
analytical research with machinery of evolutionary game theory, numerical analysis of evolutionary dynamics, and empirical work informed by such theoretical studies are all feasible.

The rest of the paper is organized as follows.
Section \ref{sec:pre} is the preliminaries,
where formulation of marriage problems, preference revelation games, and some concepts from evolutionary game theory are presented.
In the following section, we formulate the asymptotic stability of a matching.
In section \ref{sec:sim}, we consider two instances of marriage problems and present simulation results of the replicator dynamics of preference revelation games.
Section \ref{sec:disc} is the discussion, and the concluding remark is in the last section.

\section{Preliminaries} \label{sec:pre}
\subsection{Marriage problem and the preference revelation game} \label{ssec:pre-mp}
An instance of the marriage problem is denoted by a tuple $(M, W, P)$, 
where $M$ and $W$ are the sets of players on each side, and $P$ is the preference profile.
A one-to-one matching is a mapping $\mu: M \cup W \to M \cup W$.
Single players are considered to be paired with themselves.
Thus $\mu(m) \in W \cup \{m\}$ for each $m \in M$, 
$\mu(w) \in M \cup \{w\}$ for each $w \in W$, 
and if $\mu(m) = w$ then $\mu(w) = m$.
A matching $\mu$ is said to be individually rational 
if no player prefers being single to the current partner.
We say that a pair $(m, w)$ blocks a matching $\mu$
when $m$ prefers $w$ to $\mu(m)$ and $w$ prefers $m$ to $\mu(w)$.
$\mu$ is said to be stable if it is individually rational and has no blocking pair.

A matching algorithm induces the preference revelation game.
A preference revelation game is characterized by a triplet $(I, S, \pi)$, where $I$ is the set of players, $S$ is the pure strategy space, and $\pi$ is the tuple of payoff functions.
Firstly, $I = M \cup W$ in standard cases, but $I = W$ when all men are assumed to play their dominant strategies and report truthfully \cite{dubins1981, roth1982}.
As for the pure strategy space, $S^i$, the set of pure strategies for player $i \in I$, consists of $i$'s all possible preferences.
$S = \prod_{i \in I} S^i$ is the pure strategy space of the game.
Denoting the set of all possible preferences of male and female players by $S^M$ and $S^W$ respectively, 
$S^i = S^M$ for all $i \in M$ and $S^i = S^W$ for all $i \in W$.
Furthermore, a simplex 
$\Delta^i = \qty{\qty(x^i_h)_{h \in S^i} \in \mathbb{R}^{|S^i|} \relmiddle x^i_h \geq 0, \, \sum_{h \in S^i} x^i_h = 1} \subset \mathbb{R}^{|S^i|}$
denotes the mixed strategy space for player $i$, and 
$\Theta = \prod_{i \in I} \Delta^i$ 
is the mixed strategy space of the game.
Lastly, the payoff function, which maps a strategy profile to a real-valued payoff, is a composition of two mappings.
The first mapping, a matching algorithm, describes the correspondence between a strategy profile, i.e. a report profile, and the resulting matching.
The second mapping, which we call a utility function, maps each matching to a payoff value.
Formally, let $\mathcal{M}$ denote the matching space, which is the set of all possible matchings.
When a matching algorithm in use always yields the same matching from a given report profile,
the algorithm $\Gamma: S \to \mathcal{M}$ is a mapping from the pure strategy space to the matching space .
In this study, we assume that players' preferences are strict and a $M$-optimal stable algorithm is used, so that the above requirement is met.
A utility function of player $i$, $u^i: \mathcal{M} \to \mathbb{R}$, maps the resulting matching to the player's payoff.
In the current paper, we write $u^i = (u^i_1, \dots, u^i_n)$ to describe a utility function where the payoff of $i$ when matched with $j$-th preferable player is $u^i_j$.
Assuming $u^i_1 > \dots > u^i_n$, $u^i$ is considered as $i$'s cardinal preference.
Player $i$'s payoff function of the game $\pi^i: S \to \mathbb{R}$ is the composition of the matching algorithm $\Gamma$ and the utility function $u^i$: $\pi^i = u^i \circ \Gamma$.
$\pi$ in the triplet $(I, S, \pi)$ is the tuple of payoff functions: $\pi = (\pi^i)_{i \in I}$.

Our formulation of preference revelation game has two characteristics.
First, players are to report their entire preference lists only once,
so that evolutionary game theory on normal form games is applicable.
If players are to reveal their preferences one by one, the game becomes an extensive-form game.
The analysis of such situations seems possible with evolutionary game theory on extensive-form games, 
but it is beyond the scope of this paper.
Second, we consider cardinal preferences, turning preference relations into numerical payoffs with a utility function $u^i$.
This enables us to formulate and solve systems of differential equations which represent evolutionary selection dynamics.
When ordinal preferences are concerned, the payoff is determined solely by the rank of the partner in the preference list.
With cardinal preferences, in contrast, payoffs may vary among players who are matched to their $n$-th preferable partners.
As the first attempt to study evolutionary selection dynamics of the preference revelation game, we would like to have as similar conditions as possible to existing studies, many of which adopted ordinal preferences.
Therefore, we use identical $u^i$ for all players \footnote{%
Assigning different $u^i$ for each player results in different selection pressures among players.
Although such modificatoins lead to changes in specific trajectories in the strategy space, we expect that qualitative features of dynamics such as stabilities of fixed points would not be affected.
}
, which makes payoff comparison among players pointless.

\subsection{Evolutionary game dynamics} \label{ssec:pre-egt}

Replicator dynamics (RD) is one of the most significant evolutionary dynamics.
A replicator corresponds to a pure strategy, and the payoff it gains is considered as its fitness.
RD models the Darwinian natural selection, where strategies with higher payoffs prosper 
while ones with lower payoffs decline.
This model was introduced by Taylor and Jonker \cite{taylor1978} in the game theoretic context, and 
later named Replicator Dynamics by Schuster and Sigmund \cite{schuster1983}.
Researchers have actively studied RD thereafter, and 
it has also played a great role in analyzing human and social behavior \cite{cressman2014}.
Let $d = |I|$ be the number of players in the preference revelation game.
Suppose there are $d$ large and finite populations of replicators, 
each of which corresponds to a player in the game.
For instance, 
all replicators in the population corresponding to $m \in M$ have the same true preferences,
and the set of pure strategies available to them is $S^M$.
Choosing one replicator from each population randomly, 
we have $d$ replicators, each of which plays preference revelation game as the corresponding player.
Replicators earning payoffs higher than the average increase their share within their populations, 
while ones with payoffs lower than the average decline.
The dynamics of natural selection can be modeled by the repetition of this process.

Suppose that each replicator in a population of player $i \in I$ is programmed to use a fixed pure strategy $h \in S^i$.
Let $x^i_h \in [0, 1]$ denote the population share of replicators 
which always use pure strategy $h$ and whose roles are player $i$.
Then a vector $x^i = (x^i_h)_{h \in S^i}$ indicates the state of the player $i$ population.
It is formally equivalent to a mixed strategy of player $i$: i.e. $x^i \in \Delta_i$.
Let $f^i_h$ be the average payoff of the replicators of player $i$ with pure strategy $h$, 
which is equivalent to the expected payoff of $h$.
$\phi^i = \sum_{h \in S^i} x^i_h f^i_h$ is the average payoff of all player $i$ replicators, 
or the expected payoff of mixed strategy $x^i$.
Replicator dynamics of a game with $d$ player is formulated as follows:
\begin{equation} \label{eq:RD}
    \dot{x^i_h} = x^i_h \qty( f^i_h - \phi^i ) \quad \forall i \in I,
\end{equation}
where $\dot{x} = \dv{x}{t}$ and
\begin{equation}
    \sum_{h \in S^i} x^i_h = 1, \quad
    f^i_h = \sum_{s_1, \dots, s_{d-1}} x^{1}_{s_1} \dots x^{i-1}_{s_{i-1}} x^{i+1}_{s_i} \dots x^{d}_{s_{d-1}} \pi(h; s_1 \dots s_{d-1}), \quad
    \phi^i = \sum_{h \in S^i} x^i_h f^i_h.
\end{equation}
The summation in $f^i_h$ is taken for $s_1, \dots, s_{d-1}$, the strategies of players other than $i$.
In words, $x^i_h$ sums up to one for each player, $f^i_h$ is the average payoff for pure strategy $h$ of player $i$, and $\phi^i$ is the average payoff for the player $i$ population.
RD is a system of $\sum_{i \in I} (|S^i| - 1)$ differential equations.

The idea of biological reproduction is the basis for RD.
However, the selection of replicators, namely pure strategies, may well be driven by other mechanisms such as imitations, especially in the social and economic context.
We consider some classes of continuous-time selection dynamics following Weibull \cite{weibull1997}, which are generalizations of RD.

To describe the selection process, the growth-rate function is considered.
Let $g^i_h(x)$ be the growth rate of the population share corresponding to player $i$ and pure strategy $h$ under population state $x$, and
\begin{equation}
    \dot x^i_h = g^i_h(x) x^i_h \quad \forall i \in I, h \in S^i, x \in \Theta.
\end{equation}
The growth-rate function is a vector $g(x) = (g^i(x))_{i \in I}$, 
whose each component is a growth-rate function for player $i$, $g^i(x) = (g^i_h(x))_{h \in S_i}$.
Note that the growth-rate function of RD is $g^i_h(x) = f^i_h - \phi^i$.
Selection dynamics are classified based on the properties of the growth-rate function.
A regular growth-rate function is a Lipschitz continuous\footnote{
A function $f$ is Lipschitz continuous if there exists a constant $k$ 
such that $|f(x_1) - f(x_2)| < k |x_1 - x_2|$ for all $x_1, x_2$. 
} 
function $g: X \to \mathbb{R}^{|S^i|}$ 
with open domain $X \subset \mathbb{R}^{|S^i|}$ containing $\Theta$, 
such that $g^i(x) \cdot x^i = 0$ for all $x \in \Theta$ and $i \in I$
\cite[Definition 5.4]{weibull1997}.
$G$ denotes the set of regular growth-rate functions, 
and the selection dynamics induced by a regular growth-rate function will be called a regular selection dynamics on $\Theta$.
%
A regular growth-rate function $g \in G$ is payoff positive,
if $\operatorname{sgn}\qty[ g^i_h(x) ] = \operatorname{sgn}\qty[ u^i(h, x^{-i}) - u^i(x) ]$ for all $x \in \Theta, i \in I$ and $h \in S_i$,
where $\operatorname{sgn}\qty[z]$ denotes the sign of $z \in \mathbb{R}$
\cite[Definition 5.6]{weibull1997}.
$G^P \subset G$ denotes the set of payoff-positive growth-rate functions,
and the induced selection dynamics will be called a payoff-positive selection dynamics.
Note that RD of equation \eqref{eq:RD} is payoff positive.

Under the payoff-positive selection dynamics, 
a subset of the mixed strategy space $\Theta(H) \subset \Theta$ spanned by a set of pure strategy profiles $H = \prod_{i \in I} H^i \subset S$ is asymptotically stable 
if and only if $H$ is closed under weakly better replies \cite[Theorem 5.3]{weibull1997}.
If a pure strategy $h \in S^i$ earns at least as high payoff against a (mixed) strategy profile $x \in \Theta$ as $x^i$, $h$ is a weakly better reply to $x$.
Let $\alpha^i(H)$ be
\begin{equation} \label{eq:alpha}
\alpha^i(H) = \qty{h \in S^i \relmiddle u^i(h, x^{-i}) \geq u^i(x) \ \text{for some}\ x \in \Theta(H)}.
\end{equation}
$u^i(h, x^{-i})$ denotes the payoff of $h$ for player position $i$ against $x$, 
and $u^i(x)$ is the average payoff for $i$ under $x$.
$H$ is closed under weakly better replies (cuwbr) 
if $\alpha^i(H) \subset H^i$ for all $i \in I$ \cite[Definition 5.9]{weibull1997}.
The property of cuwbr is a generalized concept of a strict Nash equilibrium\footnote{
A strategy profile $x$ is a strict Nash equilibrium if $x$ is the only best response to itself.
} 
to a set.



\section{Asymptotic stability of a matching} \label{sec:asm}
In this section, we formulate asymptotic stability of a matching.
When a matching is asymptotically stable, 
it will be restored after another matching is realized due to perturbation, 
which is a change in the reporting strategy of a player.
Asymptotic stability of a matching differs from the asymptotic stability in dynamical systems in that a focal state and perturbation are not in the same space.
Whereas a focal state is on the matching space, perturbation occurs in the strategy space.
An evolutionary dynamics such as RD of equation \eqref{eq:RD} describes the time evolution of the strategy distributions in the mixed strategy space.
A matching algorithm maps a point in the mixed strategy space to a set of matchings in the matching space.
We utilize this relationship between the strategy space and the matching space, and
define asymptotic stability of a matching through the asymptotic stability of a corresponding area in the mixed strategy space, 
where the conventional definition of the asymptotic stability in dynamical systems can be used.

First, we define a partial preimage of a matching $\mu$, a subset of the pure strategy space that corresponds to $\mu$.
An obvious requirement for a partial preimage of $\mu$ is that any report profile in the set yields $\mu$ by the algorithm in use.
In addition, we require the set to be a cartesian product of the subset of the each player's pure strategy space.
This is because players determine their strategy independently in preference revelation game.
Formally, a partial preimage of a matching is defined as follows.
\begin{define} \label{def:pp}
    Let $\Gamma: S \to \mathcal{M}$ be the matching algorithm in use.
    A subset of the pure strategy space $H(\mu) \subset S$ is called a \emph{partial preimage} of matching $\mu$, when it satisfies the following:
\begin{align}
    \qq{a.} & \Gamma(r) = \mu \quad \forall \, r \in H(\mu) \\
    \qq{b.} & \exists \, \qty(H^i)_{i \in I} \qq{s.t.} 
        H^i \subset S^i \quad \forall \, i \in I, \quad H(\mu) = \prod_{i \in I} H^i.
\end{align}
\end{define}

$H(\mu)$ is "partial" in that it is included in the preimage of $\mu$ under $\Gamma$:
i.e. $H(\mu) \subset \qty{r \in S \relmiddle \Gamma(r) = \mu}$.
There may exist multiple partial preimages for a matching.
In particular, any subset of $H(\mu)$ is a trivial partial preimage of $\mu$.
In later sections, we refer to a nontrivial one that is not included in any other partial preimage of $\mu$ as a non-included partial preimage of $\mu$.

Now, we formulate asymptotic stability of a matching $\mu$ using its partial preimage $H(\mu) = \prod_{i \in I} H^i(\mu)$.
Let $\Theta(H(\mu))$ denote the subspace of the mixed strategy space spanned by $H(\mu)$:
\begin{equation}
    \Theta(H(\mu)) = \prod_{i \in I} \Delta^i(H(\mu)), \quad 
    \Delta^i(H(\mu)) = \qty{\qty(x^i_h)_{h \in H^i(\mu)} \in \mathbb{R}^{\abs{H^i(\mu)}} \relmiddle x^i_h \geq 0, \, \sum_{h \in H^i(\mu)} x^i_h = 1}.
\end{equation}
Note that any mixed report profile in $\Theta(H(\mu))$ yields only $\mu$.
Suppose that a mixed strategy profile $x$ is initially in $\Theta(H(\mu))$, and then sufficiently small perturbation drives the point $x$ out of $\Theta(H(\mu))$.
While $x$ is not in $\Theta(H(\mu))$, the algorithm may now yield matchings other than $\mu$.
If $\Theta(H(\mu))$ is asymptotically stable,
$x$ moves back into $\Theta(H(\mu))$ after some time, yielding only $\mu$.
In short, $\mu$ will be restored after sufficiently small perturbation, 
if $\Theta(H(\mu))$ is asymptotically stable.
As described in the previous section, under payoff-positive selection dynamics, 
$\Theta(H(\mu))$ is asymptotically stable if and only if $H(\mu)$ is closed under weakly better replies (cuwbr).
In this way, we formulate asymptotic stability of a matching.
\begin{define} \label{def:asm}
    A matching $\mu$ is \emph{asymptotically stable} under payoff-positive selection dynamics
    if it has a partial preimage that is closed under weakly better replies.
\end{define}

An asymptotically stable matching does not always exist, but if it does, it is also stable.
\begin{theorem}
When $M$-optimal or $W$-optimal stable algorithms are used, 
a matching that is asymptotically stable under payoff-positive selection dynamics is stable.
\end{theorem}
\begin{proof}
We concentrate on the case where $M$-optimal stable algorithms are used.
The same argument holds for $W$-optimal stable algorithms.
Players are indifferent as to which of $M$-optimal stable algorithms is in use,
because all of them yield the same matching from a given report profile.
Therefore we assume that the DA algorithm with men proposing \cite{gale1962} is used without loss of generality.

Contrary to the theorem, suppose that an unstable matching $\mu_u$ is asymptotically stable.
$\mu_u$ has a partial preimage $H(\mu_u)$ that is cuwbr.
The instability of $\mu_u$ implies that 
it has a blocking pair or it is not individually rational.
When it has a blocking pair $(m, w)$,
the fact that $m$ and $w$ are not paired with each other under $\mu_u$ indicates that $m$ did not propose to $w$.
Thus all the reports of $m$ in $H(\mu_u)$ rank $\mu_u(m)$ above $w$.
Let $P(m)$ be such a preference, where $\mu_u(m)$ is ranked above $w$,
and $P'(m)$ be a preference of $m$ where $w$ is ranked above $\mu_u(m)$.
Note that $P'(m)$ is not a element of $H(\mu_u)$.
\begin{enumerate}
    \renewcommand{\theenumi}{\roman{enumi}}%
    \item If all the reports of $w$ in $H(\mu_u)$ rank $m$ above $\mu_u(w)$, 
    $m$ can be better off by reporting $P'(m)$.
    This contradicts the closure under weakly better replies (cuwbr) of $H(\mu_u)$.
    \item If all the reports of $w$ in $H(\mu_u)$ rank $\mu_u(w)$ above $m$,
    $m$ is not worse off by reporting $P'(m)$.
    This contradicts the cuwbr of $H(\mu_u)$.
    \item If reports of $w$ in $H(\mu_u)$ include both types of preferences described above,
    $m$ is at least not worse off by reporting $P'(m)$.
    This contradicts the cuwbr of $H(\mu_u)$.
\end{enumerate}
Therefore an asymptotically stable matching cannot have a blocking pair.
When $\mu_u$ is not individually rational,
some player $i$ is paired with someone unacceptable: i.e. $i$ prefers itself to $\mu_u(i)$.
This indicates that all the reports of $i$ in $H(\mu_u)$ rank $\mu_u(i)$ above $i$,
and thus $i$ can be better off by reporting a preference where $i$ is preferred to $\mu_u(i)$,
which contradicts the cuwbr of $H(\mu_u)$.
Therefore an asymptotically stable matching must be individually rational,
and this completes the proof.
\end{proof}%
%



\section{Simulations of replicator dynamics} \label{sec:sim}
\begin{table}[t] \centering
    \caption{All the possible preferences, matchings, and their labels.}
    \label{tab:S}
    \begin{minipage}{0.25\linewidth} \centering
        \subcaption{Male's possible preferences.}
        \begin{tabular}{cc}
            \toprule
            label & preference \\ \midrule
            $M1$ & $(w_1, w_2, w_3)$ \\
            $M2$ & $(w_1, w_3, w_2)$ \\
            $M3$ & $(w_2, w_1, w_3)$ \\
            $M4$ & $(w_2, w_3, w_1)$ \\
            $M5$ & $(w_3, w_1, w_2)$ \\
            $M6$ & $(w_3, w_2, w_1)$ \\
            \bottomrule
        \end{tabular}
        \label{tab:S^M}
    \end{minipage}
    \begin{minipage}{0.25\linewidth} \centering
        \subcaption{Female's possible preferences.}
        \begin{tabular}{cc}
            \toprule
            label & preference \\ \midrule
            $W1$ & $(m_1, m_2, m_3)$ \\
            $W2$ & $(m_1, m_3, m_2)$ \\
            $W3$ & $(m_2, m_1, m_3)$ \\
            $W4$ & $(m_2, m_3, m_1)$ \\
            $W5$ & $(m_3, m_1, m_2)$ \\
            $W6$ & $(m_3, m_2, m_1)$ \\
            \bottomrule
        \end{tabular}
        \label{tab:S^W}
    \end{minipage}
    \begin{minipage}{0.4\linewidth} \centering
        \subcaption{All the possible matchings.}
        \begin{tabular}{cc}
            \toprule
            label & matching \\ \midrule
            $\mu_1$ & $[(m_1, w_1), (m_2, w_2), (m_3, w_3)]$ \\
            $\mu_2$ & $[(m_1, w_1), (m_2, w_3), (m_3, w_2)]$ \\
            $\mu_3$ & $[(m_1, w_2), (m_2, w_1), (m_3, w_3)]$ \\
            $\mu_4$ & $[(m_1, w_2), (m_2, w_3), (m_3, w_1)]$ \\
            $\mu_5$ & $[(m_1, w_3), (m_2, w_1), (m_3, w_2)]$ \\
            $\mu_6$ & $[(m_1, w_3), (m_2, w_2), (m_3, w_1)]$ \\
            \bottomrule
        \end{tabular}
        \label{tab:matching}
    \end{minipage}
\end{table}

\begin{figure}[thbp]
    \centering
    \begin{tabular}{c}
        \begin{minipage}{\linewidth}
            \centering
            \includegraphics[width=0.9\linewidth]{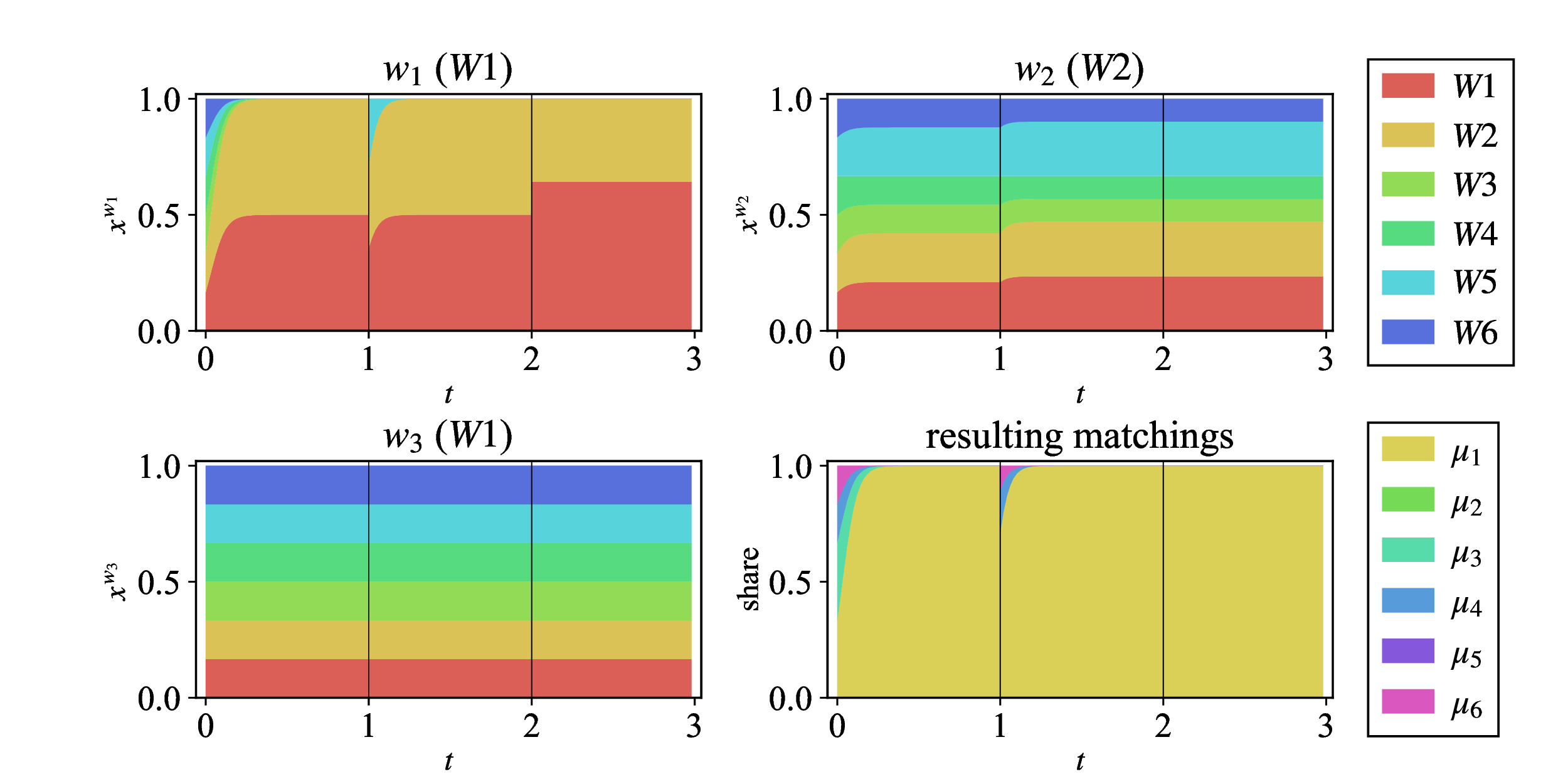}
            \subcaption{%
                The dynamics of the first problem $(M, W, P_1)$.
                0.1 was added to $x^{w_1}_{W5}$ at $t = 1$ and to $x^{w_1}_{W1}$ at $t = 2$ as perturbations.
                The state got stationary at $t = 1, 2, 3$.
                At the stationary state, $W1$ and $W2$ survived in the $w_1$ population, and no strategy went extinct in the $w_2$ and $w_3$ populations.
                There were 72 possible report profiles at the stationary state.
            }
            \label{fig:simres-P1}
        \end{minipage}
        \\
        \begin{minipage}{\linewidth}
            \centering
            \includegraphics[width=0.9\linewidth]{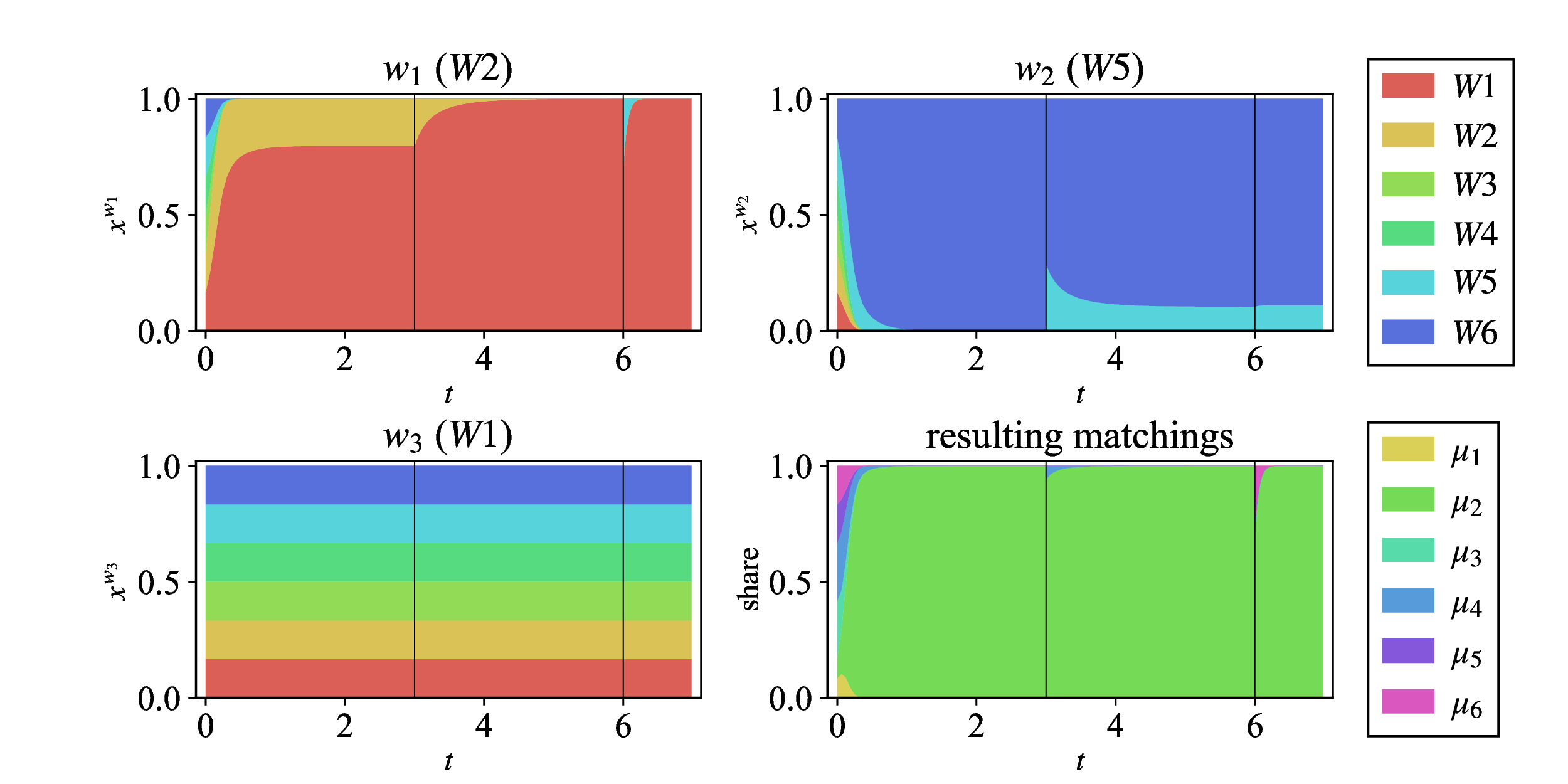}
            \subcaption{
                The dynamics of the second problem $(M, W, P_2)$.
                0.1 was added to $x^{w_2}_{W5}$ at $t = 3$ and to $x^{w_1}_{W5}$ at $t = 6$ as perturbations.
                The state got stationary at $t = 3, 6, 7$.
                At the stationary state, $W1$ and $W2$ in the $w_1$ population and $W6$ in the $w_2$ population survived.
                No strategy went extinct in the $w_3$ population.
                There were 12 possible report profiles at the stationary state.
            }
            \label{fig:simres-P2}
        \end{minipage}
    \end{tabular}
    \caption{%
        RD of equation \eqref{eq:RD} were numerically solved from the initial condition
        $x^i_0 = \left( \frac{1}{6} \right)_{h \in S^i}$ for all $i \in I$.
        The utility function was set to $u_{\tconv} = (25, 10, 1)$.
        The player position $i$ and its true preference are shown on top of each plot.
        The lower right panel shows the shares of the resulting matchings at every time.
        The other three panels show the time evolutions of the strategy profile in each player population.
    }
    \label{fig:simres}
\end{figure}

So far, we have discussed game dynamics of preference revelation games theoretically.
In this section, we take two instances marriage problems, and present numerical simulations of replicator dynamics (RD) of preference revelation games.
We pay particular attention to matchings that will be realized after report profiles evolve for some time according to RD.
We consider marriage problems with three players on each side: i.e. $M = \qty{m_1, m_2, m_3}$ and $W = \qty{w_1, w_2, w_3}$.
It is assumed that male players do not use weakly dominated strategies, which means they always report truthfully.
Thus in this section, the set of strategic players is $I = W$.
All the possible reports are shown in tables \ref{tab:S},
where $(m_1, m_2, m_3)$ denotes a preference of a woman who prefers $m_1$ to $m_2$, and $m_2$ to $m_3$.
We will refer to a pure strategy by its label shown in the tables.
The pure strategy space of the game is $S = \qty{W1, W2, \dots, W6}^3$.
It is assumed in this section that all players are acceptable.
RD of equation \eqref{eq:RD} were numerically solved with the initial condition where all pure strategies evenly exist: i.e.
$x^i_0 = \qty( \frac{1}{6} )_{h \in S^i}$ for all $i \in I$.
The DA algorithm with men proposing was used in obtaining a matching from a report profile, and 
matchings were mapped to payoffs by the utility function $u_{\tconv} = (25, 10, 1)$.
For example, when a replicator is paired with the most preferable man, its payoff is 25.
We say the state $x$ is at the stationary state at time $t$, if the maximum changes in $x^i$ during the last few calculation steps were sufficiently small\footnote{%
Let the state at time $t$, which corresponds to the $n$-th calculation point, be $x_n$.
Let the state at $i$ calculation steps before $x_n$ be $x_{n - i}$.
We say $x$ is at the stationary at time $t$, if $\max_{i \in I}(\abs{x^i_{n - i} - x^i_{n - i - 1}}) < 10^{-5}$ for $i \in \qty{0, \dots, 3}$.
}.
Furthermore, we say that a pure strategy survived when its population share at the stationary state was more than 0.1\%.

The marriage problems we investigated are $(M, W, P_1)$ and $(M, W, P_2)$, where 
\begin{equation}
    P_1 = (M1, M1, M2, W1, W2, W1) \qq{and} P_2 = (M3, M1, M1, W2, W5, W1).
\end{equation}
The preference of $m_1$ is shown first in a preference profile, 
that of $m_2$ second and so on, the last being the preference of $w_3$.
In the first problem with $P_1$, the $M$-optimal and $W$-optimal stable matchings are
\begin{equation}
    \mu^M_1 = [(m_1, w_1), (m_2, w_2), (m_3, w_3)] \qq{and} \mu^W_1 = [(m_1, w_1), (m_2, w_3), (m_3, w_2)].
\end{equation}
Under $\mu^M_1$, $m_1$ is paired with $w_1$, $m_2$ paired with $w_2$, and $m_3$ paired with $w_3$.
The $M$-optimal and $W$-optimal stable matchings do not coincide, 
but no woman has the incentive to falsify assuming others are truthful\footnote{%
    Theorem 1 of Gale and Sotomeyer \cite{gale1985a} asserts that 
    there is at least one woman who will be better off by falsifying, if there is more than one stable matching.
    In the problem with $P_1$, no woman has the incentive to misreport even though there are multiple stable matchings,
    because we assumed that all players were acceptable.
}.
In the second problem\footnote{%
    $P_2$ was taken from the proof of Theorem 3 in Roth \cite{roth1982}.
} with $P_2$, the $M$-optimal and $W$-optimal stable matchings are
\begin{equation}
    \mu^M_2 = [(m_1, w_2), (m_2, w_3), (m_3, w_1)] \qq{and} \mu^W_2 = [(m_1, w_1), (m_2, w_3), (m_3, w_2)].
\end{equation}
Under $M$-optimal stable algorithms, at least one woman has the incentive to falsify so that $\mu^W_2$ instead of $\mu^M_2$ is obtained.
For instance, $\mu^W_2$ is realized if $w_2$ falsely reports $W6$.

Before presenting the simulation results, we check if asymptotically stable matchings exist in the two marriage problems\footnote{
    In practice, we first simulated replicator dynamics and then studied asymptotically stable matchings referring to the simulation results.
}.
In the first problem with $P_1$, the following set is a non-included partial preimage of $\mu^M_1$ under $M$-optimal stable algorithms:
\begin{equation} \label{eq:H(muM1)}
    H(\mu^M_1) := \{W1, W2\} \times \{W1, \dots, W6\} \times \{W1, \dots, W6\}.
\end{equation}
$H(\mu^M_1)$ turns out to be cuwbr, and thus $\mu^M_1$ is asymptotically stable, as shown in \ref{ssec:proof-p1}.
In the proof, we essentially checked if $\alpha(H(\mu^M_1))$ of equation \eqref{eq:alpha} satisfied the definition of cuwbr with payoff tables,
which was made workable due to the assumption of $I = W$.
In the second problem with $P_2$, one finds two non-included partial preimages of $\mu^W_2$ under $M$-optimal stable algorithms:
\begin{align}
    H_1(\mu^W_2) &= \{W1, W2\} \times \{W6\} \times \{W1 \dots W6\},
    \label{eq:H1(muW2)} \\
    H_2(\mu^W_2) &= \{W1\} \times \{W5, W6\} \times \{W1 \dots W6\}.
    \label{eq:H2(muW2)}
\end{align}
The following relations hold for $H_1$ and $H_2$, where $H^i$ satisfies $H = \prod_{i \in I} H^i$:
\begin{equation} \label{eq:cuwbr-H1}
    \begin{aligned}
        \alpha^{w_1}(H_1) &\subset H^{w_1}_1 \\
        \alpha^{w_2}(H_1) &\subset H^{w_2}_1 \cup \{W5\} \not\subset H^{w_2}_1
    \end{aligned}
\end{equation}
\begin{equation} \label{eq:cuwbr-H2}
    \begin{aligned}
        \alpha^{w_1}(H_2) &\subset H^{w_1}_2 \cup \{W2\} \not\subset H^{w_1}_2 \\
        \alpha^{w_2}(H_2) &\subset H^{w_2}_2
    \end{aligned}
\end{equation}
Equations \eqref{eq:cuwbr-H1} and \eqref{eq:cuwbr-H2} indicate that neither $H_1$ nor $H_2$ is cuwbr.
Indeed, one can show that there is no partial preimage of $\mu^W_2$ that is cuwbr, implying that $\mu^W_2$ is not asymptotically stable.
Details on the proofs are in \ref{ssec:proof-p2} and \ref{ssec:proof-p2-nocuwbr}.
In addition, there is no partial preimage of $\mu^M_2$ that is cuwbr, since $w_1$'s reporting $W1$ earns a higher payoff when others are reporting truthfully.
Therefore, no asymptotically stable matching exists in the second problem.

In figures \ref{fig:simres}, the time evolution of $x^i$ and the resulting matchings are presented.
In the first problem (figure \ref{fig:simres-P1}), truthful reportings of all players have survived.
All the possible report profiles at the stationary state result in the $M$-optimal stable matching $\mu^M_1$.
One can see that although the values of some $x^i_h$ changed after perturbation (e.g. $x^{w_1}_{W1}$ at $t = 2$), the set of surviving reports were always $H(\mu^M_1)$, and the resulting matching remained $\mu^M_1$ regardless of the perturbations.
This demonstrates the asymptotic stability of $H(\mu^M_1)$ and thus $\mu^M_1$.
In the second problem (figure \ref{fig:simres-P2}), truthful reporting went extinct in the $w_2$ population.
All the possible report profiles at the stationary state result in the $W$-optimal stable matching $\mu^W_2$, not $\mu^M_2$.
However, while the set of surviving reports was $H_1(\mu^W_2)$ at $t = 3$, it changed to $H_2(\mu^W_2)$ as the result of the perturbation at $t = 3$.
In short, even though $\mu^W_2$ seemed robust against the perturbations, it is not asymptotically stable in our definition.
We will explore this situation in the next section.

To conclude this section, we briefly discuss the simulation results.
Firstly, the selection favored only one matching in both cases, even though several strategies coexisted in many player populations at the stationary state.
This is not due to the choice of the preference profiles.
Assuming that three players are present on each side and all players are acceptable, there are $6^6$ possible preference profiles.
We fixed the preference of $m_1$ without loss of generality, 
and conducted the simulation of RD for all $6^5$ possible preference profiles in case of $I = W$.
Except the 18 cases (0.2\%) where $x$ kept fluctuating and thus we could not determine the final states,
only one matching was possible at the final state in all instances.
These observations imply that it is natural to consider the asymptotic stability of a matching.
Secondly, the simulation results were consistent with previous research on the incentives of players.
In the first problem, where no woman has the incentive to misreport, the truthful reporting strategies of women survived.
In the second problem, at least one woman has the incentive to falsify, and only the false reporting survived in the $w_2$ population accordingly.
Still, $w_2$ truthfully reported the most preferable men, which is consistent with the result of Roth \cite{roth1982} that no woman has the incentive to falsify her most preferred man.

\section{Discussion} \label{sec:disc}

\begin{figure}[thbp]
    \centering
    \begin{tabular}{c}
        \begin{minipage}{\linewidth}
            \centering
            \includegraphics[width=0.9\linewidth]{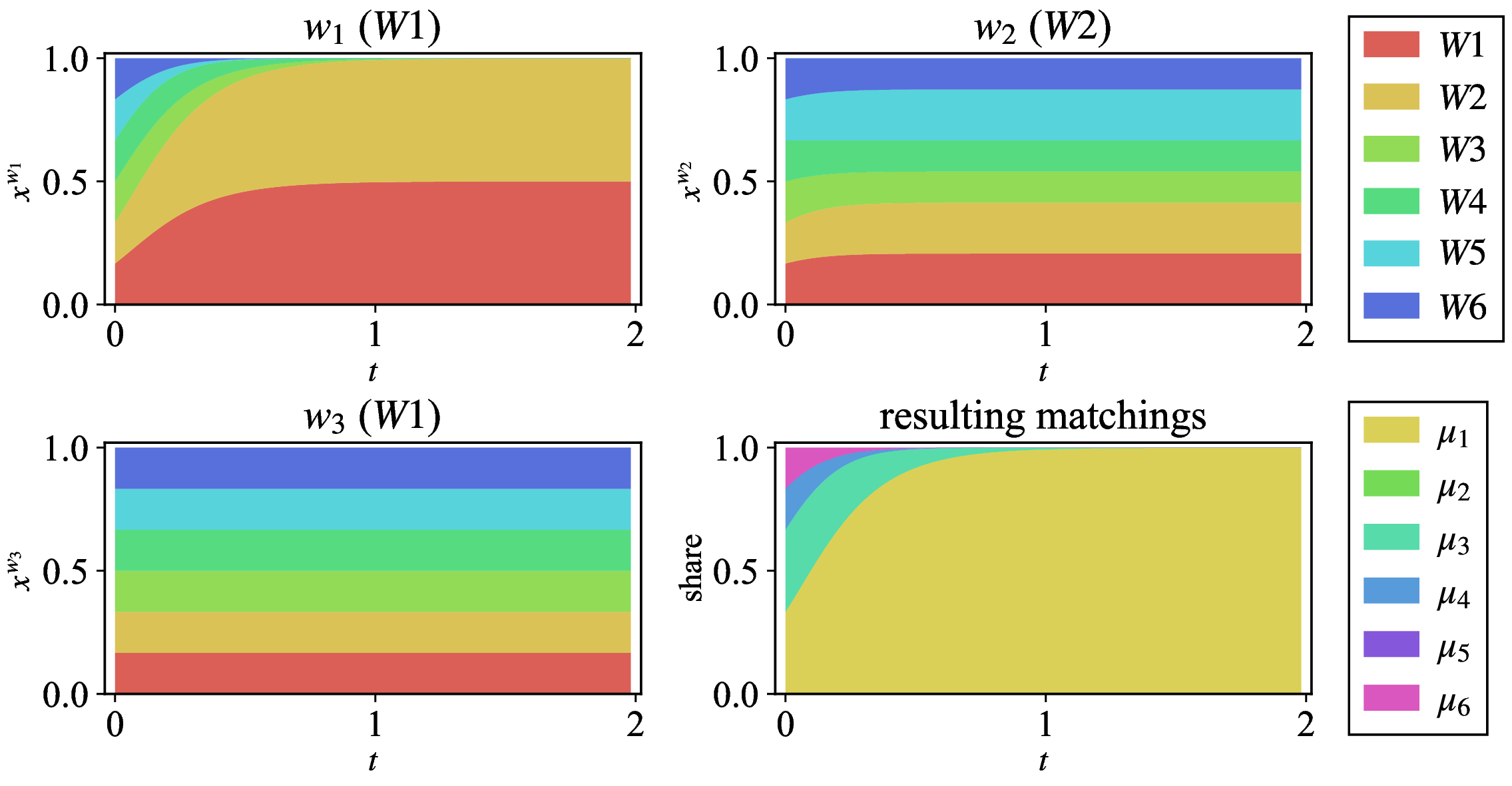}
            \subcaption{%
                The time evolution of $x^i$ with $P_1$.
                $u_{\tlin} = (15, 10, 5)$ was used as the utility function.
                The dynamics of $x^i$ look similar with figure \ref{fig:simres-P1},
                but the time scale was slower with $u_{\tlin}$ than with $u_{\tconv}$.
            }
            \label{fig:simres-P1l}
        \end{minipage}
        \\
        \begin{minipage}{\linewidth}
            \centering
            \includegraphics[width=.9\linewidth]{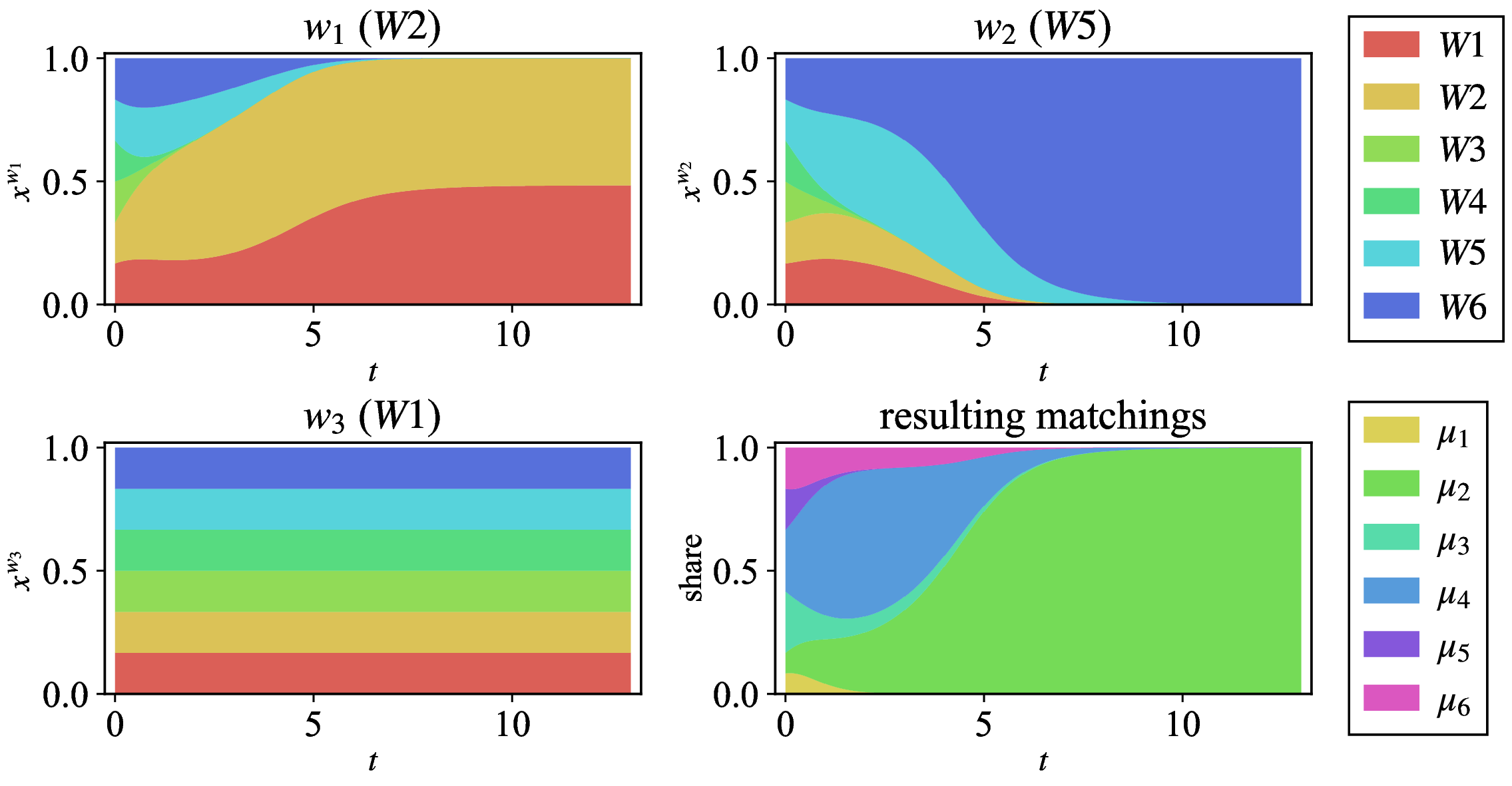}
            \subcaption{%
                The time evolution of $x^i$ with $P_2$.
                $u_{\tconc} = (5\sqrt{3}, 5\sqrt{2}, 5)$ was used as the utility function.
                Not only the transient dynamics but also the share of $W1$ in the $w_1$ population at the stationary state differed from figure \ref{fig:simres-P2}, where $u_{\tconv}$ was used.
            }
            \label{fig:simres-P2sq}
        \end{minipage}
        \end{tabular}
    \caption{%
        RD of equation \eqref{eq:RD} were numerically solved from the initial condition
        $x^i_0 = \qty(\frac{1}{6})_{h \in S^i}$  for all $i \in I$.
        The player position $i$ and its true preference are shown on top of each plot.
        The lower right panel shows the shares of the resulting matchings at every time.
        The other three panels show the time evolutions of the strategy profile in each player population.
    }
\end{figure}

The simulation result of the second problem $(M, W, P_2)$ in figure \ref{fig:simres-P2} suggests that
the $W$-optimal stable matching $\mu^W_2$ was restored after small perturbations were given, even though $\mu^W_2$ is not asymptotically stable.
What equations \eqref{eq:cuwbr-H1} and \eqref{eq:cuwbr-H2} tell us about two non-included partial preimages of $\mu^W_2$ are that
$H_1$ is not cuwbr because $W5$ of $w_2$ is a weakly better reply to some $x \in \Theta(H_1)$, and that
$H_2$ is not cuwbr because $W2$ of $w_1$ is a weakly better reply to some $x \in \Theta(H_2)$.
Suppose that the current population state is $x \in \Theta(H_1)$, and some replicators in $w_2$ population start to use $W5$.
This would move the population state to $x' \in \Theta(H')$, where
\begin{equation}
    H' := \{W1, W2\} \times \{W5, W6\} \times \{W1, \dots, W6\}.
\end{equation}
If we start from $x \in \Theta(H_2)$ and some replicators in $w_1$ population come to use $W2$, the population state would also get $x' \in \Theta(H')$.
In addition, one can show that $x' \in \Theta(H')$ evolves into another state $x'' \in \Theta(H_1 \cup H_2)$, where only $\mu^W_2$ is obtained, as long as selection dynamics are payoff-positive\footnote{%
    Details are in \ref{sec:proof-disc}.
}.
Therefore $\mu^W_2$ will in fact be restored after small perturbations are given.

In the second problem with $P_2$, $\mu^W_2$ is not asymptotically stable, even though it will be restored after small perturbations.
The essentials behind this confusing situation are two non-included partial preimages of $\mu^W_2$, one of which contains a weakly better reply to the other and vice versa.
All partial preimages are not cuwbr and thus $\mu^W_2$ is not asymptotically stable by our definition,
but still, a mixed strategy profile $x$ moves only within partial preimages of the same matching.
In order to characterize the robustness of $\mu^W_2$, we propose a weaker property of a matching, quasi-asymptotic stability:
\begin{define} \label{def:qasm}
    A matching $\mu$ is \emph{quasi-asymptotically stable} 
    if it has a set of partial preimages
    \begin{equation}
        \qty{H_1(\mu), \dots, H_n(\mu)} \qq{such that} \alpha^i(H_m(\mu)) \subset \bigcup_{k = 1}^{n} H^i_k(\mu) \quad \forall\ m \in \{1, \dots, n\}, \ i \in I.
    \end{equation}
\end{define}
When $n = 1$, it is equivalent to the asymptotic stability of Definition \ref{def:asm}.
$\mu^W_2$ is quasi-asymptotically stable with $\qty{H_1, H_2}$.
We presume the quasi-asymptotic stability portrays some kind of dynamical stability of a matching similar to the asymptotic stability,
but we have not yet been able to neither prove it nor find a counterexample.

In the current paper, preference relations are described by numerical values.
Obviously, the payoff values must satisfy $u^i_1 > \dots > u^i_n$, but still 
there are countless candidates for the utility function.
To examine whether the choice of the utility function affects the dynamics,
we conducted simulations of RD in the two marriage problems with two other utility functions, 
$u_{\tlin} = (15, 10, 5)$ and $u_{\tconc} = (5\sqrt{3}, 5\sqrt{2}, 5)$.
The utility function used in the preceding section was $u_{\tconv} = (25, 10, 1)$.
$u_{\tconv}$, $u_{\tlin}$ and $u_{\tconc}$ are convex, linear and concave, respectively.
In the first problem with $P_1$, the dynamics of $x^i$ did not seem to be affected by the choice of the utility function, but their time scales differed.
The dynamics with the linear utility function is shown in figure \ref{fig:simres-P1l}.
In the second problem with $P_2$, transient dynamics slightly varied and the share of $W1$ in the $w_1$ population at the last state was also affected by the choice of the utility function.
Figure \ref{fig:simres-P2sq} shows the dynamics with the concave utility function.
Nevertheless, strategies that survived remained the same regardless of the utility function.
Our analysis have mainly focused on the last states of dynamics and the matchings which were obtained there.
Thus, in view of the observations above, we expect our simulation results to be independent on the choice of the utility function.

\section{Conclusion} \label{sec:conc}
Although the literature on centralized matching markets often assumes that a true preference of each player is known to herself and fixed throughout the matching process, several empirical results \cite{dwenger2018, narita2018,grenet2021} cast doubt on the assumption.
To circumvent the problem, evolutionary dynamics of preference revelation games are considered.
We formulated the asymptotic stability of a matching, which indicates the dynamical robustness of the matching against sufficiently small changes in players' reporting strategies.
We showed that the asymptotic stability of a matching implies its stability, thereby contributing a practical suggestion on how to obtain stable matchings in centralized markets:
by having a learning phase during which players find reporting strategies that suit their true preferences through trial and error, an asymptotically stable and hence stable matching would be obtained.
Note that our insight does not assume players' initial awareness of their true preferences, unlike much of the literature.

We believe the current paper lays a foundation for dynamical approach to centralized matching markets with evolutionary dynamics of preference revelation games,
and there is plenty of room for future research.
One issue that should be tackled is the difficulty in finding asymptotically stable matchings.
We checked the asymptotic stability according to its definition, looking at payoff matrices, but 
this procedure quickly becomes unfeasible as the number of players increase.
Our approach with continuous-time deterministic selection dynamics is applicable to wider range of problems, 
provided the payoff function is deterministic and thus payoffs are never determined randomly.
We assumed preferences were strict and $M$-optimal stable algorithms were used, but
players' preferences need not be strict. 
When preferences are not strict, deterministic payoff functions can be constructed by employing matching algorithms that have deterministic tie-breaking rules.
Our approach can also be adopted to many-to-one or many-to-many matching problems.
In these cases, payoff functions are likely to be much complicated and straightforward analysis may get cumbersome.
Directions for future research include analyses of other types of game dynamics, such as ones with mutations or on graphs, and 
search for connections between the asymptotic stability of a matching and stochastic stability analyzed in research on decentralized markets \cite{jackson2002,newton2015,bilancini2020}.
Another interesting topic would be a partial preimage of a matching.
We introduced a partial preimage to formulate the asymptotic stability of a matching based on the concept,
but a partial preimage itself may hold information that characterizes the matching.
For instance, the cardinality of a non-included partial preimage is likely to reflect the likelihood that a matching is realized when all players randomly submit their reports.


    \appendix
\section{Payoff matrices}
\begin{table}
    \centering
    \caption{The payoff matrices when the preference profile is $P_1$, $I = W$, 
            and the utility function is $\bm{u}_{conv} = (25, 5, 1)$.
            Each subtable shows the payoff matrix when the strategy of $w_3$ is fixed.
            The row and column labels indicate the strategy of $w_1$ and $w_2$, respectively.
            Each entry denotes the payoffs for $w_1$, $w_2$, and $w_3$ from the top.}
    \label{tab:pm-P1}
    \begin{subtable}{0.31\linewidth}
        \centering
        \caption{The payoffs when $w_3$ uses $W1$.}
        \label{tab:pm-P1-W1}
        \scriptsize
        \begin{tabular}{cc@{ }c@{ }c@{ }c@{ }c@{ }c} \toprule
            & \multicolumn{6}{c}{$w_2$} \\ \cmidrule{2-7}
            $w_1$ & $W1$ & $W2$ & $W3$ & $W4$ & $W5$ & $W6$ \\ \toprule
            $W1$ &
                $\begin{matrix} 25 \\ 1 \\ 1 \end{matrix}$ &
                $\begin{matrix} 25 \\ 1 \\ 1 \end{matrix}$ &
                $\begin{matrix} 25 \\ 1 \\ 1 \end{matrix}$ &
                $\begin{matrix} 25 \\ 1 \\ 1 \end{matrix}$ &
                $\begin{matrix} 25 \\ 1 \\ 1 \end{matrix}$ &
                $\begin{matrix} 25 \\ 1 \\ 1 \end{matrix}$ \\ \midrule
            $W2$ &
                $\begin{matrix} 25 \\ 1 \\ 1 \end{matrix}$ &
                $\begin{matrix} 25 \\ 1 \\ 1 \end{matrix}$ &
                $\begin{matrix} 25 \\ 1 \\ 1 \end{matrix}$ &
                $\begin{matrix} 25 \\ 1 \\ 1 \end{matrix}$ &
                $\begin{matrix} 25 \\ 1 \\ 1 \end{matrix}$ &
                $\begin{matrix} 25 \\ 1 \\ 1 \end{matrix}$ \\ \midrule
            $W3$ &
                $\begin{matrix} 5 \\ 25 \\ 1 \end{matrix}$ &
                $\begin{matrix} 5 \\ 25 \\ 1 \end{matrix}$ &
                $\begin{matrix} 5 \\ 25 \\ 1 \end{matrix}$ &
                $\begin{matrix} 5 \\ 25 \\ 1 \end{matrix}$ &
                $\begin{matrix} 5 \\ 25 \\ 1 \end{matrix}$ &
                $\begin{matrix} 5 \\ 25 \\ 1 \end{matrix}$ \\ \midrule
            $W4$ &
                $\begin{matrix} 5 \\ 25 \\ 1 \end{matrix}$ &
                $\begin{matrix} 5 \\ 25 \\ 1 \end{matrix}$ &
                $\begin{matrix} 5 \\ 25 \\ 1 \end{matrix}$ &
                $\begin{matrix} 5 \\ 25 \\ 1 \end{matrix}$ &
                $\begin{matrix} 5 \\ 25 \\ 1 \end{matrix}$ &
                $\begin{matrix} 5 \\ 25 \\ 1 \end{matrix}$ \\ \midrule
            $W5$ &
                $\begin{matrix} 1 \\ 25 \\ 5 \end{matrix}$ &
                $\begin{matrix} 1 \\ 25 \\ 5 \end{matrix}$ &
                $\begin{matrix} 1 \\ 1 \\ 25 \end{matrix}$ &
                $\begin{matrix} 1 \\ 1 \\ 25 \end{matrix}$ &
                $\begin{matrix} 1 \\ 25 \\ 5 \end{matrix}$ &
                $\begin{matrix} 1 \\ 1 \\ 25 \end{matrix}$ \\ \midrule
            $W6$ &
                $\begin{matrix} 1 \\ 25 \\ 5 \end{matrix}$ &
                $\begin{matrix} 1 \\ 25 \\ 5 \end{matrix}$ &
                $\begin{matrix} 1 \\ 1 \\ 25 \end{matrix}$ &
                $\begin{matrix} 1 \\ 1 \\ 25 \end{matrix}$ &
                $\begin{matrix} 1 \\ 25 \\ 5 \end{matrix}$ &
                $\begin{matrix} 1 \\ 1 \\ 25 \end{matrix}$ \\ \bottomrule
        \end{tabular}
    \end{subtable}
    \begin{subtable}{0.31\linewidth}
        \centering
        \caption{The payoffs when $w_3$ uses $W2$.}
        \label{tab:pm-P1-W2}
        \scriptsize
        \begin{tabular}{cc@{ }c@{ }c@{ }c@{ }c@{ }c} \toprule
            & \multicolumn{6}{c}{$w_2$} \\ \cmidrule{2-7}
            $w_1$ & $W1$ & $W2$ & $W3$ & $W4$ & $W5$ & $W6$ \\ \toprule
            $W1$ &
                $\begin{matrix} 25 \\ 1 \\ 1 \end{matrix}$ &
                $\begin{matrix} 25 \\ 1 \\ 1 \end{matrix}$ &
                $\begin{matrix} 25 \\ 1 \\ 1 \end{matrix}$ &
                $\begin{matrix} 25 \\ 1 \\ 1 \end{matrix}$ &
                $\begin{matrix} 25 \\ 1 \\ 1 \end{matrix}$ &
                $\begin{matrix} 25 \\ 1 \\ 1 \end{matrix}$ \\ \midrule
            $W2$ &
                $\begin{matrix} 25 \\ 1 \\ 1 \end{matrix}$ &
                $\begin{matrix} 25 \\ 1 \\ 1 \end{matrix}$ &
                $\begin{matrix} 25 \\ 1 \\ 1 \end{matrix}$ &
                $\begin{matrix} 25 \\ 1 \\ 1 \end{matrix}$ &
                $\begin{matrix} 25 \\ 1 \\ 1 \end{matrix}$ &
                $\begin{matrix} 25 \\ 1 \\ 1 \end{matrix}$ \\ \midrule
            $W3$ &
                $\begin{matrix} 5 \\ 25 \\ 1 \end{matrix}$ &
                $\begin{matrix} 5 \\ 25 \\ 1 \end{matrix}$ &
                $\begin{matrix} 5 \\ 25 \\ 1 \end{matrix}$ &
                $\begin{matrix} 5 \\ 25 \\ 1 \end{matrix}$ &
                $\begin{matrix} 5 \\ 25 \\ 1 \end{matrix}$ &
                $\begin{matrix} 5 \\ 25 \\ 1 \end{matrix}$ \\ \midrule
            $W4$ &
                $\begin{matrix} 5 \\ 25 \\ 1 \end{matrix}$ &
                $\begin{matrix} 5 \\ 25 \\ 1 \end{matrix}$ &
                $\begin{matrix} 5 \\ 25 \\ 1 \end{matrix}$ &
                $\begin{matrix} 5 \\ 25 \\ 1 \end{matrix}$ &
                $\begin{matrix} 5 \\ 25 \\ 1 \end{matrix}$ &
                $\begin{matrix} 5 \\ 25 \\ 1 \end{matrix}$ \\ \midrule
            $W5$ &
                $\begin{matrix} 1 \\ 25 \\ 5 \end{matrix}$ &
                $\begin{matrix} 1 \\ 25 \\ 5 \end{matrix}$ &
                $\begin{matrix} 1 \\ 1 \\ 25 \end{matrix}$ &
                $\begin{matrix} 1 \\ 1 \\ 25 \end{matrix}$ &
                $\begin{matrix} 1 \\ 25 \\ 5 \end{matrix}$ &
                $\begin{matrix} 1 \\ 1 \\ 25 \end{matrix}$ \\ \midrule
            $W6$ &
                $\begin{matrix} 1 \\ 25 \\ 5 \end{matrix}$ &
                $\begin{matrix} 1 \\ 25 \\ 5 \end{matrix}$ &
                $\begin{matrix} 1 \\ 1 \\ 25 \end{matrix}$ &
                $\begin{matrix} 1 \\ 1 \\ 25 \end{matrix}$ &
                $\begin{matrix} 1 \\ 25 \\ 5 \end{matrix}$ &
                $\begin{matrix} 1 \\ 1 \\ 25 \end{matrix}$ \\ \bottomrule
        \end{tabular}
    \end{subtable}
    \begin{subtable}{0.31\linewidth}
        \centering
        \caption{The payoffs when $w_3$ uses $W3$.}
        \label{tab:pm-P1-W3}
        \scriptsize
        \begin{tabular}{cc@{ }c@{ }c@{ }c@{ }c@{ }c} \toprule
            & \multicolumn{6}{c}{$w_2$} \\ \cmidrule{2-7}
            $w_1$ & $W1$ & $W2$ & $W3$ & $W4$ & $W5$ & $W6$ \\ \toprule
            $W1$ &
                $\begin{matrix} 25 \\ 1 \\ 1 \end{matrix}$ &
                $\begin{matrix} 25 \\ 1 \\ 1 \end{matrix}$ &
                $\begin{matrix} 25 \\ 1 \\ 1 \end{matrix}$ &
                $\begin{matrix} 25 \\ 1 \\ 1 \end{matrix}$ &
                $\begin{matrix} 25 \\ 1 \\ 1 \end{matrix}$ &
                $\begin{matrix} 25 \\ 1 \\ 1 \end{matrix}$ \\ \midrule
            $W2$ &
                $\begin{matrix} 25 \\ 1 \\ 1 \end{matrix}$ &
                $\begin{matrix} 25 \\ 1 \\ 1 \end{matrix}$ &
                $\begin{matrix} 25 \\ 1 \\ 1 \end{matrix}$ &
                $\begin{matrix} 25 \\ 1 \\ 1 \end{matrix}$ &
                $\begin{matrix} 25 \\ 1 \\ 1 \end{matrix}$ &
                $\begin{matrix} 25 \\ 1 \\ 1 \end{matrix}$ \\ \midrule
            $W3$ &
                $\begin{matrix} 5 \\ 25 \\ 1 \end{matrix}$ &
                $\begin{matrix} 5 \\ 25 \\ 1 \end{matrix}$ &
                $\begin{matrix} 5 \\ 25 \\ 1 \end{matrix}$ &
                $\begin{matrix} 5 \\ 25 \\ 1 \end{matrix}$ &
                $\begin{matrix} 5 \\ 25 \\ 1 \end{matrix}$ &
                $\begin{matrix} 5 \\ 25 \\ 1 \end{matrix}$ \\ \midrule
            $W4$ &
                $\begin{matrix} 5 \\ 25 \\ 1 \end{matrix}$ &
                $\begin{matrix} 5 \\ 25 \\ 1 \end{matrix}$ &
                $\begin{matrix} 5 \\ 25 \\ 1 \end{matrix}$ &
                $\begin{matrix} 5 \\ 25 \\ 1 \end{matrix}$ &
                $\begin{matrix} 5 \\ 25 \\ 1 \end{matrix}$ &
                $\begin{matrix} 5 \\ 25 \\ 1 \end{matrix}$ \\ \midrule
            $W5$ &
                $\begin{matrix} 1 \\ 25 \\ 5 \end{matrix}$ &
                $\begin{matrix} 1 \\ 25 \\ 5 \end{matrix}$ &
                $\begin{matrix} 1 \\ 1 \\ 25 \end{matrix}$ &
                $\begin{matrix} 1 \\ 1 \\ 25 \end{matrix}$ &
                $\begin{matrix} 1 \\ 25 \\ 5 \end{matrix}$ &
                $\begin{matrix} 1 \\ 1 \\ 25 \end{matrix}$ \\ \midrule
            $W6$ &
                $\begin{matrix} 1 \\ 25 \\ 5 \end{matrix}$ &
                $\begin{matrix} 1 \\ 25 \\ 5 \end{matrix}$ &
                $\begin{matrix} 1 \\ 1 \\ 25 \end{matrix}$ &
                $\begin{matrix} 1 \\ 1 \\ 25 \end{matrix}$ &
                $\begin{matrix} 1 \\ 25 \\ 5 \end{matrix}$ &
                $\begin{matrix} 1 \\ 1 \\ 25 \end{matrix}$ \\ \bottomrule
        \end{tabular}
    \end{subtable}
    \bigskip

    \begin{subtable}{0.31\linewidth}
        \centering
        \caption{The payoffs when $w_3$ uses $W4$.}
        \label{tab:pm-P1-W4}
        \scriptsize
        \begin{tabular}{cc@{ }c@{ }c@{ }c@{ }c@{ }c} \toprule
            & \multicolumn{6}{c}{$w_2$} \\ \cmidrule{2-7}
            $w_1$ & $W1$ & $W2$ & $W3$ & $W4$ & $W5$ & $W6$ \\ \toprule
            $W1$ &
                $\begin{matrix} 25 \\ 1 \\ 1 \end{matrix}$ &
                $\begin{matrix} 25 \\ 1 \\ 1 \end{matrix}$ &
                $\begin{matrix} 25 \\ 1 \\ 1 \end{matrix}$ &
                $\begin{matrix} 25 \\ 1 \\ 1 \end{matrix}$ &
                $\begin{matrix} 25 \\ 1 \\ 1 \end{matrix}$ &
                $\begin{matrix} 25 \\ 1 \\ 1 \end{matrix}$ \\ \midrule
            $W2$ &
                $\begin{matrix} 25 \\ 1 \\ 1 \end{matrix}$ &
                $\begin{matrix} 25 \\ 1 \\ 1 \end{matrix}$ &
                $\begin{matrix} 25 \\ 1 \\ 1 \end{matrix}$ &
                $\begin{matrix} 25 \\ 1 \\ 1 \end{matrix}$ &
                $\begin{matrix} 25 \\ 1 \\ 1 \end{matrix}$ &
                $\begin{matrix} 25 \\ 1 \\ 1 \end{matrix}$ \\ \midrule
            $W3$ &
                $\begin{matrix} 5 \\ 25 \\ 1 \end{matrix}$ &
                $\begin{matrix} 5 \\ 25 \\ 1 \end{matrix}$ &
                $\begin{matrix} 5 \\ 25 \\ 1 \end{matrix}$ &
                $\begin{matrix} 5 \\ 25 \\ 1 \end{matrix}$ &
                $\begin{matrix} 5 \\ 25 \\ 1 \end{matrix}$ &
                $\begin{matrix} 5 \\ 25 \\ 1 \end{matrix}$ \\ \midrule
            $W4$ &
                $\begin{matrix} 5 \\ 25 \\ 1 \end{matrix}$ &
                $\begin{matrix} 5 \\ 25 \\ 1 \end{matrix}$ &
                $\begin{matrix} 5 \\ 25 \\ 1 \end{matrix}$ &
                $\begin{matrix} 5 \\ 25 \\ 1 \end{matrix}$ &
                $\begin{matrix} 5 \\ 25 \\ 1 \end{matrix}$ &
                $\begin{matrix} 5 \\ 25 \\ 1 \end{matrix}$ \\ \midrule
            $W5$ &
                $\begin{matrix} 1 \\ 25 \\ 5 \end{matrix}$ &
                $\begin{matrix} 1 \\ 25 \\ 5 \end{matrix}$ &
                $\begin{matrix} 1 \\ 1 \\ 25 \end{matrix}$ &
                $\begin{matrix} 1 \\ 1 \\ 25 \end{matrix}$ &
                $\begin{matrix} 1 \\ 25 \\ 5 \end{matrix}$ &
                $\begin{matrix} 1 \\ 1 \\ 25 \end{matrix}$ \\ \midrule
            $W6$ &
                $\begin{matrix} 1 \\ 25 \\ 5 \end{matrix}$ &
                $\begin{matrix} 1 \\ 25 \\ 5 \end{matrix}$ &
                $\begin{matrix} 1 \\ 1 \\ 25 \end{matrix}$ &
                $\begin{matrix} 1 \\ 1 \\ 25 \end{matrix}$ &
                $\begin{matrix} 1 \\ 25 \\ 5 \end{matrix}$ &
                $\begin{matrix} 1 \\ 1 \\ 25 \end{matrix}$ \\ \bottomrule
        \end{tabular}
    \end{subtable}
    \begin{subtable}{0.31\linewidth}
        \centering
        \caption{The payoffs when $w_3$ uses $W5$.}
        \label{tab:pm-P1-W5}
        \scriptsize
        \begin{tabular}{cc@{ }c@{ }c@{ }c@{ }c@{ }c} \toprule
            & \multicolumn{6}{c}{$w_2$} \\ \cmidrule{2-7}
            $w_1$ & $W1$ & $W2$ & $W3$ & $W4$ & $W5$ & $W6$ \\ \toprule
            $W1$ &
                $\begin{matrix} 25 \\ 1 \\ 1 \end{matrix}$ &
                $\begin{matrix} 25 \\ 1 \\ 1 \end{matrix}$ &
                $\begin{matrix} 25 \\ 1 \\ 1 \end{matrix}$ &
                $\begin{matrix} 25 \\ 1 \\ 1 \end{matrix}$ &
                $\begin{matrix} 25 \\ 1 \\ 1 \end{matrix}$ &
                $\begin{matrix} 25 \\ 1 \\ 1 \end{matrix}$ \\ \midrule
            $W2$ &
                $\begin{matrix} 25 \\ 1 \\ 1 \end{matrix}$ &
                $\begin{matrix} 25 \\ 1 \\ 1 \end{matrix}$ &
                $\begin{matrix} 25 \\ 1 \\ 1 \end{matrix}$ &
                $\begin{matrix} 25 \\ 1 \\ 1 \end{matrix}$ &
                $\begin{matrix} 25 \\ 1 \\ 1 \end{matrix}$ &
                $\begin{matrix} 25 \\ 1 \\ 1 \end{matrix}$ \\ \midrule
            $W3$ &
                $\begin{matrix} 5 \\ 25 \\ 1 \end{matrix}$ &
                $\begin{matrix} 5 \\ 25 \\ 1 \end{matrix}$ &
                $\begin{matrix} 5 \\ 25 \\ 1 \end{matrix}$ &
                $\begin{matrix} 5 \\ 25 \\ 1 \end{matrix}$ &
                $\begin{matrix} 5 \\ 25 \\ 1 \end{matrix}$ &
                $\begin{matrix} 5 \\ 25 \\ 1 \end{matrix}$ \\ \midrule
            $W4$ &
                $\begin{matrix} 5 \\ 25 \\ 1 \end{matrix}$ &
                $\begin{matrix} 5 \\ 25 \\ 1 \end{matrix}$ &
                $\begin{matrix} 5 \\ 25 \\ 1 \end{matrix}$ &
                $\begin{matrix} 5 \\ 25 \\ 1 \end{matrix}$ &
                $\begin{matrix} 5 \\ 25 \\ 1 \end{matrix}$ &
                $\begin{matrix} 5 \\ 25 \\ 1 \end{matrix}$ \\ \midrule
            $W5$ &
                $\begin{matrix} 1 \\ 25 \\ 5 \end{matrix}$ &
                $\begin{matrix} 1 \\ 25 \\ 5 \end{matrix}$ &
                $\begin{matrix} 1 \\ 1 \\ 25 \end{matrix}$ &
                $\begin{matrix} 1 \\ 1 \\ 25 \end{matrix}$ &
                $\begin{matrix} 1 \\ 25 \\ 5 \end{matrix}$ &
                $\begin{matrix} 1 \\ 1 \\ 25 \end{matrix}$ \\ \midrule
            $W6$ &
                $\begin{matrix} 1 \\ 25 \\ 5 \end{matrix}$ &
                $\begin{matrix} 1 \\ 25 \\ 5 \end{matrix}$ &
                $\begin{matrix} 1 \\ 1 \\ 25 \end{matrix}$ &
                $\begin{matrix} 1 \\ 1 \\ 25 \end{matrix}$ &
                $\begin{matrix} 1 \\ 25 \\ 5 \end{matrix}$ &
                $\begin{matrix} 1 \\ 1 \\ 25 \end{matrix}$ \\ \bottomrule
        \end{tabular}
    \end{subtable}
    \begin{subtable}{0.31\linewidth}
        \centering
        \caption{The payoffs when $w_3$ uses $W6$.}
        \label{tab:pm-P1-W6}
        \scriptsize
        \begin{tabular}{cc@{ }c@{ }c@{ }c@{ }c@{ }c} \toprule
            & \multicolumn{6}{c}{$w_2$} \\ \cmidrule{2-7}
            $w_1$ & $W1$ & $W2$ & $W3$ & $W4$ & $W5$ & $W6$ \\ \toprule
            $W1$ &
                $\begin{matrix} 25 \\ 1 \\ 1 \end{matrix}$ &
                $\begin{matrix} 25 \\ 1 \\ 1 \end{matrix}$ &
                $\begin{matrix} 25 \\ 1 \\ 1 \end{matrix}$ &
                $\begin{matrix} 25 \\ 1 \\ 1 \end{matrix}$ &
                $\begin{matrix} 25 \\ 1 \\ 1 \end{matrix}$ &
                $\begin{matrix} 25 \\ 1 \\ 1 \end{matrix}$ \\ \midrule
            $W2$ &
                $\begin{matrix} 25 \\ 1 \\ 1 \end{matrix}$ &
                $\begin{matrix} 25 \\ 1 \\ 1 \end{matrix}$ &
                $\begin{matrix} 25 \\ 1 \\ 1 \end{matrix}$ &
                $\begin{matrix} 25 \\ 1 \\ 1 \end{matrix}$ &
                $\begin{matrix} 25 \\ 1 \\ 1 \end{matrix}$ &
                $\begin{matrix} 25 \\ 1 \\ 1 \end{matrix}$ \\ \midrule
            $W3$ &
                $\begin{matrix} 5 \\ 25 \\ 1 \end{matrix}$ &
                $\begin{matrix} 5 \\ 25 \\ 1 \end{matrix}$ &
                $\begin{matrix} 5 \\ 25 \\ 1 \end{matrix}$ &
                $\begin{matrix} 5 \\ 25 \\ 1 \end{matrix}$ &
                $\begin{matrix} 5 \\ 25 \\ 1 \end{matrix}$ &
                $\begin{matrix} 5 \\ 25 \\ 1 \end{matrix}$ \\ \midrule
            $W4$ &
                $\begin{matrix} 5 \\ 25 \\ 1 \end{matrix}$ &
                $\begin{matrix} 5 \\ 25 \\ 1 \end{matrix}$ &
                $\begin{matrix} 5 \\ 25 \\ 1 \end{matrix}$ &
                $\begin{matrix} 5 \\ 25 \\ 1 \end{matrix}$ &
                $\begin{matrix} 5 \\ 25 \\ 1 \end{matrix}$ &
                $\begin{matrix} 5 \\ 25 \\ 1 \end{matrix}$ \\ \midrule
            $W5$ &
                $\begin{matrix} 1 \\ 25 \\ 5 \end{matrix}$ &
                $\begin{matrix} 1 \\ 25 \\ 5 \end{matrix}$ &
                $\begin{matrix} 1 \\ 1 \\ 25 \end{matrix}$ &
                $\begin{matrix} 1 \\ 1 \\ 25 \end{matrix}$ &
                $\begin{matrix} 1 \\ 25 \\ 5 \end{matrix}$ &
                $\begin{matrix} 1 \\ 1 \\ 25 \end{matrix}$ \\ \midrule
            $W6$ &
                $\begin{matrix} 1 \\ 25 \\ 5 \end{matrix}$ &
                $\begin{matrix} 1 \\ 25 \\ 5 \end{matrix}$ &
                $\begin{matrix} 1 \\ 1 \\ 25 \end{matrix}$ &
                $\begin{matrix} 1 \\ 1 \\ 25 \end{matrix}$ &
                $\begin{matrix} 1 \\ 25 \\ 5 \end{matrix}$ &
                $\begin{matrix} 1 \\ 1 \\ 25 \end{matrix}$ \\ \bottomrule
        \end{tabular}
    \end{subtable}
\end{table}
\begin{table}
    \centering
    \caption{The payoff matrices when the preference profile is $P_2$, $I = W$, 
            and the utility function is $\bm{u}_{conv} = (25, 5, 1)$.
            Each subtable shows the payoff matrix when the strategy of $w_3$ is fixed.
            The row and column labels indicate the strategy of $w_1$ and $w_2$, respectively.
            Each entry denotes the payoffs for $w_1$, $w_2$, and $w_3$ from the top.}
    \label{tab:pm-P2}
    \begin{subtable}{0.31\linewidth}
        \centering
        \caption{The payoffs when $w_3$ uses $W1$.}
        \label{tab:pm-P2-W1}
        \scriptsize
        \begin{tabular}{cc@{ }c@{ }c@{ }c@{ }c@{ }c} \toprule
            & \multicolumn{6}{c}{$w_2$} \\ \cmidrule{2-7}
            $w_1$ & $W1$ & $W2$ & $W3$ & $W4$ & $W5$ & $W6$ \\ \toprule
            $W1$ &
                $\begin{matrix} 1 \\ 5 \\ 1 \end{matrix}$ &
                $\begin{matrix} 1 \\ 5 \\ 1 \end{matrix}$ &
                $\begin{matrix} 1 \\ 5 \\ 1 \end{matrix}$ &
                $\begin{matrix} 25 \\ 1 \\ 1 \end{matrix}$ &
                $\begin{matrix} 25 \\ 25 \\ 5 \end{matrix}$ &
                $\begin{matrix} 25 \\ 25 \\ 5 \end{matrix}$ \\ \midrule
            $W2$ &
                $\begin{matrix} 5 \\ 5 \\ 5 \end{matrix}$ &
                $\begin{matrix} 5 \\ 5 \\ 5 \end{matrix}$ &
                $\begin{matrix} 25 \\ 1 \\ 1 \end{matrix}$ &
                $\begin{matrix} 25 \\ 1 \\ 1 \end{matrix}$ &
                $\begin{matrix} 5 \\ 5 \\ 5 \end{matrix}$ &
                $\begin{matrix} 25 \\ 25 \\ 5 \end{matrix}$ \\ \midrule
            $W3$ &
                $\begin{matrix} 1 \\ 5 \\ 1 \end{matrix}$ &
                $\begin{matrix} 1 \\ 5 \\ 1 \end{matrix}$ &
                $\begin{matrix} 1 \\ 5 \\ 1 \end{matrix}$ &
                $\begin{matrix} 1 \\ 25 \\ 25 \end{matrix}$ &
                $\begin{matrix} 1 \\ 25 \\ 25 \end{matrix}$ &
                $\begin{matrix} 1 \\ 25 \\ 25 \end{matrix}$ \\ \midrule
            $W4$ &
                $\begin{matrix} 1 \\ 5 \\ 1 \end{matrix}$ &
                $\begin{matrix} 1 \\ 5 \\ 1 \end{matrix}$ &
                $\begin{matrix} 1 \\ 5 \\ 1 \end{matrix}$ &
                $\begin{matrix} 1 \\ 25 \\ 25 \end{matrix}$ &
                $\begin{matrix} 1 \\ 25 \\ 25 \end{matrix}$ &
                $\begin{matrix} 1 \\ 25 \\ 25 \end{matrix}$ \\ \midrule
            $W5$ &
                $\begin{matrix} 5 \\ 5 \\ 5 \end{matrix}$ &
                $\begin{matrix} 5 \\ 5 \\ 5 \end{matrix}$ &
                $\begin{matrix} 5 \\ 1 \\ 25 \end{matrix}$ &
                $\begin{matrix} 5 \\ 1 \\ 25 \end{matrix}$ &
                $\begin{matrix} 5 \\ 5 \\ 5 \end{matrix}$ &
                $\begin{matrix} 5 \\ 1 \\ 25 \end{matrix}$ \\ \midrule
            $W6$ &
                $\begin{matrix} 5 \\ 5 \\ 5 \end{matrix}$ &
                $\begin{matrix} 5 \\ 5 \\ 5 \end{matrix}$ &
                $\begin{matrix} 5 \\ 1 \\ 25 \end{matrix}$ &
                $\begin{matrix} 5 \\ 1 \\ 25 \end{matrix}$ &
                $\begin{matrix} 5 \\ 5 \\ 5 \end{matrix}$ &
                $\begin{matrix} 5 \\ 1 \\ 25 \end{matrix}$ \\ \bottomrule
        \end{tabular}
    \end{subtable}
    \begin{subtable}{0.31\linewidth}
        \centering
        \caption{The payoffs when $w_3$ uses $W2$.}
        \label{tab:pm-P2-W2}
        \scriptsize
        \begin{tabular}{cc@{ }c@{ }c@{ }c@{ }c@{ }c} \toprule
            & \multicolumn{6}{c}{$w_2$} \\ \cmidrule{2-7}
            $w_1$ & $W1$ & $W2$ & $W3$ & $W4$ & $W5$ & $W6$ \\ \toprule
            $W1$ &
                $\begin{matrix} 1 \\ 5 \\ 1 \end{matrix}$ &
                $\begin{matrix} 1 \\ 5 \\ 1 \end{matrix}$ &
                $\begin{matrix} 1 \\ 5 \\ 1 \end{matrix}$ &
                $\begin{matrix} 25 \\ 1 \\ 1 \end{matrix}$ &
                $\begin{matrix} 25 \\ 25 \\ 5 \end{matrix}$ &
                $\begin{matrix} 25 \\ 25 \\ 5 \end{matrix}$ \\ \midrule
            $W2$ &
                $\begin{matrix} 5 \\ 5 \\ 5 \end{matrix}$ &
                $\begin{matrix} 5 \\ 5 \\ 5 \end{matrix}$ &
                $\begin{matrix} 25 \\ 1 \\ 1 \end{matrix}$ &
                $\begin{matrix} 25 \\ 1 \\ 1 \end{matrix}$ &
                $\begin{matrix} 5 \\ 5 \\ 5 \end{matrix}$ &
                $\begin{matrix} 25 \\ 25 \\ 5 \end{matrix}$ \\ \midrule
            $W3$ &
                $\begin{matrix} 1 \\ 5 \\ 1 \end{matrix}$ &
                $\begin{matrix} 1 \\ 5 \\ 1 \end{matrix}$ &
                $\begin{matrix} 1 \\ 5 \\ 1 \end{matrix}$ &
                $\begin{matrix} 1 \\ 25 \\ 25 \end{matrix}$ &
                $\begin{matrix} 1 \\ 25 \\ 25 \end{matrix}$ &
                $\begin{matrix} 1 \\ 25 \\ 25 \end{matrix}$ \\ \midrule
            $W4$ &
                $\begin{matrix} 1 \\ 5 \\ 1 \end{matrix}$ &
                $\begin{matrix} 1 \\ 5 \\ 1 \end{matrix}$ &
                $\begin{matrix} 1 \\ 5 \\ 1 \end{matrix}$ &
                $\begin{matrix} 1 \\ 25 \\ 25 \end{matrix}$ &
                $\begin{matrix} 1 \\ 25 \\ 25 \end{matrix}$ &
                $\begin{matrix} 1 \\ 25 \\ 25 \end{matrix}$ \\ \midrule
            $W5$ &
                $\begin{matrix} 5 \\ 5 \\ 5 \end{matrix}$ &
                $\begin{matrix} 5 \\ 5 \\ 5 \end{matrix}$ &
                $\begin{matrix} 5 \\ 1 \\ 25 \end{matrix}$ &
                $\begin{matrix} 5 \\ 1 \\ 25 \end{matrix}$ &
                $\begin{matrix} 5 \\ 5 \\ 5 \end{matrix}$ &
                $\begin{matrix} 5 \\ 1 \\ 25 \end{matrix}$ \\ \midrule
            $W6$ &
                $\begin{matrix} 5 \\ 5 \\ 5 \end{matrix}$ &
                $\begin{matrix} 5 \\ 5 \\ 5 \end{matrix}$ &
                $\begin{matrix} 5 \\ 1 \\ 25 \end{matrix}$ &
                $\begin{matrix} 5 \\ 1 \\ 25 \end{matrix}$ &
                $\begin{matrix} 5 \\ 5 \\ 5 \end{matrix}$ &
                $\begin{matrix} 5 \\ 1 \\ 25 \end{matrix}$ \\ \bottomrule
        \end{tabular}
    \end{subtable}
    \begin{subtable}{0.31\linewidth}
        \centering
        \caption{The payoffs when $w_3$ uses $W3$.}
        \label{tab:pm-P2-W3}
        \scriptsize
        \begin{tabular}{cc@{ }c@{ }c@{ }c@{ }c@{ }c} \toprule
            & \multicolumn{6}{c}{$w_2$} \\ \cmidrule{2-7}
            $w_1$ & $W1$ & $W2$ & $W3$ & $W4$ & $W5$ & $W6$ \\ \toprule
            $W1$ &
                $\begin{matrix} 1 \\ 5 \\ 1 \end{matrix}$ &
                $\begin{matrix} 1 \\ 5 \\ 1 \end{matrix}$ &
                $\begin{matrix} 1 \\ 5 \\ 1 \end{matrix}$ &
                $\begin{matrix} 25 \\ 1 \\ 1 \end{matrix}$ &
                $\begin{matrix} 25 \\ 25 \\ 5 \end{matrix}$ &
                $\begin{matrix} 25 \\ 25 \\ 5 \end{matrix}$ \\ \midrule
            $W2$ &
                $\begin{matrix} 5 \\ 5 \\ 5 \end{matrix}$ &
                $\begin{matrix} 5 \\ 5 \\ 5 \end{matrix}$ &
                $\begin{matrix} 25 \\ 1 \\ 1 \end{matrix}$ &
                $\begin{matrix} 25 \\ 1 \\ 1 \end{matrix}$ &
                $\begin{matrix} 5 \\ 5 \\ 5 \end{matrix}$ &
                $\begin{matrix} 25 \\ 25 \\ 5 \end{matrix}$ \\ \midrule
            $W3$ &
                $\begin{matrix} 1 \\ 5 \\ 1 \end{matrix}$ &
                $\begin{matrix} 1 \\ 5 \\ 1 \end{matrix}$ &
                $\begin{matrix} 1 \\ 5 \\ 1 \end{matrix}$ &
                $\begin{matrix} 1 \\ 25 \\ 25 \end{matrix}$ &
                $\begin{matrix} 1 \\ 25 \\ 25 \end{matrix}$ &
                $\begin{matrix} 1 \\ 25 \\ 25 \end{matrix}$ \\ \midrule
            $W4$ &
                $\begin{matrix} 1 \\ 5 \\ 1 \end{matrix}$ &
                $\begin{matrix} 1 \\ 5 \\ 1 \end{matrix}$ &
                $\begin{matrix} 1 \\ 5 \\ 1 \end{matrix}$ &
                $\begin{matrix} 1 \\ 25 \\ 25 \end{matrix}$ &
                $\begin{matrix} 1 \\ 25 \\ 25 \end{matrix}$ &
                $\begin{matrix} 1 \\ 25 \\ 25 \end{matrix}$ \\ \midrule
            $W5$ &
                $\begin{matrix} 5 \\ 5 \\ 5 \end{matrix}$ &
                $\begin{matrix} 5 \\ 5 \\ 5 \end{matrix}$ &
                $\begin{matrix} 5 \\ 1 \\ 25 \end{matrix}$ &
                $\begin{matrix} 5 \\ 1 \\ 25 \end{matrix}$ &
                $\begin{matrix} 5 \\ 5 \\ 5 \end{matrix}$ &
                $\begin{matrix} 5 \\ 1 \\ 25 \end{matrix}$ \\ \midrule
            $W6$ &
                $\begin{matrix} 5 \\ 5 \\ 5 \end{matrix}$ &
                $\begin{matrix} 5 \\ 5 \\ 5 \end{matrix}$ &
                $\begin{matrix} 5 \\ 1 \\ 25 \end{matrix}$ &
                $\begin{matrix} 5 \\ 1 \\ 25 \end{matrix}$ &
                $\begin{matrix} 5 \\ 5 \\ 5 \end{matrix}$ &
                $\begin{matrix} 5 \\ 1 \\ 25 \end{matrix}$ \\ \bottomrule
        \end{tabular}
    \end{subtable}
    \bigskip

    \begin{subtable}{0.31\linewidth}
        \centering
        \caption{The payoffs when $w_3$ uses $W4$.}
        \label{tab:pm-P2-W4}
        \scriptsize
        \begin{tabular}{cc@{ }c@{ }c@{ }c@{ }c@{ }c} \toprule
            & \multicolumn{6}{c}{$w_2$} \\ \cmidrule{2-7}
            $w_1$ & $W1$ & $W2$ & $W3$ & $W4$ & $W5$ & $W6$ \\ \toprule
            $W1$ &
                $\begin{matrix} 1 \\ 5 \\ 1 \end{matrix}$ &
                $\begin{matrix} 1 \\ 5 \\ 1 \end{matrix}$ &
                $\begin{matrix} 1 \\ 5 \\ 1 \end{matrix}$ &
                $\begin{matrix} 25 \\ 1 \\ 1 \end{matrix}$ &
                $\begin{matrix} 25 \\ 25 \\ 5 \end{matrix}$ &
                $\begin{matrix} 25 \\ 25 \\ 5 \end{matrix}$ \\ \midrule
            $W2$ &
                $\begin{matrix} 5 \\ 5 \\ 5 \end{matrix}$ &
                $\begin{matrix} 5 \\ 5 \\ 5 \end{matrix}$ &
                $\begin{matrix} 25 \\ 1 \\ 1 \end{matrix}$ &
                $\begin{matrix} 25 \\ 1 \\ 1 \end{matrix}$ &
                $\begin{matrix} 5 \\ 5 \\ 5 \end{matrix}$ &
                $\begin{matrix} 25 \\ 25 \\ 5 \end{matrix}$ \\ \midrule
            $W3$ &
                $\begin{matrix} 1 \\ 5 \\ 1 \end{matrix}$ &
                $\begin{matrix} 1 \\ 5 \\ 1 \end{matrix}$ &
                $\begin{matrix} 1 \\ 5 \\ 1 \end{matrix}$ &
                $\begin{matrix} 1 \\ 25 \\ 25 \end{matrix}$ &
                $\begin{matrix} 1 \\ 25 \\ 25 \end{matrix}$ &
                $\begin{matrix} 1 \\ 25 \\ 25 \end{matrix}$ \\ \midrule
            $W4$ &
                $\begin{matrix} 1 \\ 5 \\ 1 \end{matrix}$ &
                $\begin{matrix} 1 \\ 5 \\ 1 \end{matrix}$ &
                $\begin{matrix} 1 \\ 5 \\ 1 \end{matrix}$ &
                $\begin{matrix} 1 \\ 25 \\ 25 \end{matrix}$ &
                $\begin{matrix} 1 \\ 25 \\ 25 \end{matrix}$ &
                $\begin{matrix} 1 \\ 25 \\ 25 \end{matrix}$ \\ \midrule
            $W5$ &
                $\begin{matrix} 5 \\ 5 \\ 5 \end{matrix}$ &
                $\begin{matrix} 5 \\ 5 \\ 5 \end{matrix}$ &
                $\begin{matrix} 5 \\ 1 \\ 25 \end{matrix}$ &
                $\begin{matrix} 5 \\ 1 \\ 25 \end{matrix}$ &
                $\begin{matrix} 5 \\ 5 \\ 5 \end{matrix}$ &
                $\begin{matrix} 5 \\ 1 \\ 25 \end{matrix}$ \\ \midrule
            $W6$ &
                $\begin{matrix} 5 \\ 5 \\ 5 \end{matrix}$ &
                $\begin{matrix} 5 \\ 5 \\ 5 \end{matrix}$ &
                $\begin{matrix} 5 \\ 1 \\ 25 \end{matrix}$ &
                $\begin{matrix} 5 \\ 1 \\ 25 \end{matrix}$ &
                $\begin{matrix} 5 \\ 5 \\ 5 \end{matrix}$ &
                $\begin{matrix} 5 \\ 1 \\ 25 \end{matrix}$ \\ \bottomrule
        \end{tabular}
    \end{subtable}
    \begin{subtable}{0.31\linewidth}
        \centering
        \caption{The payoffs when $w_3$ uses $W5$.}
        \label{tab:pm-P2-W5}
        \scriptsize
        \begin{tabular}{cc@{ }c@{ }c@{ }c@{ }c@{ }c} \toprule
            & \multicolumn{6}{c}{$w_2$} \\ \cmidrule{2-7}
            $w_1$ & $W1$ & $W2$ & $W3$ & $W4$ & $W5$ & $W6$ \\ \toprule
            $W1$ &
                $\begin{matrix} 1 \\ 5 \\ 1 \end{matrix}$ &
                $\begin{matrix} 1 \\ 5 \\ 1 \end{matrix}$ &
                $\begin{matrix} 1 \\ 5 \\ 1 \end{matrix}$ &
                $\begin{matrix} 25 \\ 1 \\ 1 \end{matrix}$ &
                $\begin{matrix} 25 \\ 25 \\ 5 \end{matrix}$ &
                $\begin{matrix} 25 \\ 25 \\ 5 \end{matrix}$ \\ \midrule
            $W2$ &
                $\begin{matrix} 5 \\ 5 \\ 5 \end{matrix}$ &
                $\begin{matrix} 5 \\ 5 \\ 5 \end{matrix}$ &
                $\begin{matrix} 25 \\ 1 \\ 1 \end{matrix}$ &
                $\begin{matrix} 25 \\ 1 \\ 1 \end{matrix}$ &
                $\begin{matrix} 5 \\ 5 \\ 5 \end{matrix}$ &
                $\begin{matrix} 25 \\ 25 \\ 5 \end{matrix}$ \\ \midrule
            $W3$ &
                $\begin{matrix} 1 \\ 5 \\ 1 \end{matrix}$ &
                $\begin{matrix} 1 \\ 5 \\ 1 \end{matrix}$ &
                $\begin{matrix} 1 \\ 5 \\ 1 \end{matrix}$ &
                $\begin{matrix} 1 \\ 25 \\ 25 \end{matrix}$ &
                $\begin{matrix} 1 \\ 25 \\ 25 \end{matrix}$ &
                $\begin{matrix} 1 \\ 25 \\ 25 \end{matrix}$ \\ \midrule
            $W4$ &
                $\begin{matrix} 1 \\ 5 \\ 1 \end{matrix}$ &
                $\begin{matrix} 1 \\ 5 \\ 1 \end{matrix}$ &
                $\begin{matrix} 1 \\ 5 \\ 1 \end{matrix}$ &
                $\begin{matrix} 1 \\ 25 \\ 25 \end{matrix}$ &
                $\begin{matrix} 1 \\ 25 \\ 25 \end{matrix}$ &
                $\begin{matrix} 1 \\ 25 \\ 25 \end{matrix}$ \\ \midrule
            $W5$ &
                $\begin{matrix} 5 \\ 5 \\ 5 \end{matrix}$ &
                $\begin{matrix} 5 \\ 5 \\ 5 \end{matrix}$ &
                $\begin{matrix} 5 \\ 1 \\ 25 \end{matrix}$ &
                $\begin{matrix} 5 \\ 1 \\ 25 \end{matrix}$ &
                $\begin{matrix} 5 \\ 5 \\ 5 \end{matrix}$ &
                $\begin{matrix} 5 \\ 1 \\ 25 \end{matrix}$ \\ \midrule
            $W6$ &
                $\begin{matrix} 5 \\ 5 \\ 5 \end{matrix}$ &
                $\begin{matrix} 5 \\ 5 \\ 5 \end{matrix}$ &
                $\begin{matrix} 5 \\ 1 \\ 25 \end{matrix}$ &
                $\begin{matrix} 5 \\ 1 \\ 25 \end{matrix}$ &
                $\begin{matrix} 5 \\ 5 \\ 5 \end{matrix}$ &
                $\begin{matrix} 5 \\ 1 \\ 25 \end{matrix}$ \\ \bottomrule
        \end{tabular}
    \end{subtable}
    \begin{subtable}{0.31\linewidth}
        \centering
        \caption{The payoffs when $w_3$ uses $W6$.}
        \label{tab:pm-P2-W6}
        \scriptsize
        \begin{tabular}{cc@{ }c@{ }c@{ }c@{ }c@{ }c} \toprule
            & \multicolumn{6}{c}{$w_2$} \\ \cmidrule{2-7}
            $w_1$ & $W1$ & $W2$ & $W3$ & $W4$ & $W5$ & $W6$ \\ \toprule
            $W1$ &
                $\begin{matrix} 1 \\ 5 \\ 1 \end{matrix}$ &
                $\begin{matrix} 1 \\ 5 \\ 1 \end{matrix}$ &
                $\begin{matrix} 1 \\ 5 \\ 1 \end{matrix}$ &
                $\begin{matrix} 25 \\ 1 \\ 1 \end{matrix}$ &
                $\begin{matrix} 25 \\ 25 \\ 5 \end{matrix}$ &
                $\begin{matrix} 25 \\ 25 \\ 5 \end{matrix}$ \\ \midrule
            $W2$ &
                $\begin{matrix} 5 \\ 5 \\ 5 \end{matrix}$ &
                $\begin{matrix} 5 \\ 5 \\ 5 \end{matrix}$ &
                $\begin{matrix} 25 \\ 1 \\ 1 \end{matrix}$ &
                $\begin{matrix} 25 \\ 1 \\ 1 \end{matrix}$ &
                $\begin{matrix} 5 \\ 5 \\ 5 \end{matrix}$ &
                $\begin{matrix} 25 \\ 25 \\ 5 \end{matrix}$ \\ \midrule
            $W3$ &
                $\begin{matrix} 1 \\ 5 \\ 1 \end{matrix}$ &
                $\begin{matrix} 1 \\ 5 \\ 1 \end{matrix}$ &
                $\begin{matrix} 1 \\ 5 \\ 1 \end{matrix}$ &
                $\begin{matrix} 1 \\ 25 \\ 25 \end{matrix}$ &
                $\begin{matrix} 1 \\ 25 \\ 25 \end{matrix}$ &
                $\begin{matrix} 1 \\ 25 \\ 25 \end{matrix}$ \\ \midrule
            $W4$ &
                $\begin{matrix} 1 \\ 5 \\ 1 \end{matrix}$ &
                $\begin{matrix} 1 \\ 5 \\ 1 \end{matrix}$ &
                $\begin{matrix} 1 \\ 5 \\ 1 \end{matrix}$ &
                $\begin{matrix} 1 \\ 25 \\ 25 \end{matrix}$ &
                $\begin{matrix} 1 \\ 25 \\ 25 \end{matrix}$ &
                $\begin{matrix} 1 \\ 25 \\ 25 \end{matrix}$ \\ \midrule
            $W5$ &
                $\begin{matrix} 5 \\ 5 \\ 5 \end{matrix}$ &
                $\begin{matrix} 5 \\ 5 \\ 5 \end{matrix}$ &
                $\begin{matrix} 5 \\ 1 \\ 25 \end{matrix}$ &
                $\begin{matrix} 5 \\ 1 \\ 25 \end{matrix}$ &
                $\begin{matrix} 5 \\ 5 \\ 5 \end{matrix}$ &
                $\begin{matrix} 5 \\ 1 \\ 25 \end{matrix}$ \\ \midrule
            $W6$ &
                $\begin{matrix} 5 \\ 5 \\ 5 \end{matrix}$ &
                $\begin{matrix} 5 \\ 5 \\ 5 \end{matrix}$ &
                $\begin{matrix} 5 \\ 1 \\ 25 \end{matrix}$ &
                $\begin{matrix} 5 \\ 1 \\ 25 \end{matrix}$ &
                $\begin{matrix} 5 \\ 5 \\ 5 \end{matrix}$ &
                $\begin{matrix} 5 \\ 1 \\ 25 \end{matrix}$ \\ \bottomrule
        \end{tabular}
    \end{subtable}
\end{table}

The payoff matrices of the preference revelation games analyzed in section \ref{sec:sim} are presented in tables \ref{tab:pm-P1} and \ref{tab:pm-P2}.
The utility function is set to $\bm{u}_{\tconv} = (25, 10, 1)$.
We assume that men always report truthfully: i.e. $I = W$.

\section{Proofs for section \ref{sec:sim}}
\subsection{Closure under weakly better replies of $H(\mu^M_1)$} \label{ssec:proof-p1}
We show that $H(\mu^M_1)$ of equation \eqref{eq:H(muM1)} is cuwbr.
\begin{proof}
    \newcommand{\tin}{\text{in}}
    \newcommand{\tout}{\text{out}}
    First, recall that $H = \prod_{i \in I} H^i$ is cuwbr if $\alpha^i(H) \subset H^i$ for all $i \in I$, 
    where
    \begin{equation}
        \alpha^i(H) = \{ h \in S^i \mid \exists \, x \in \Theta(H) \ \text{s.t.} \ u^i(h, x^{-i}) \geq u^i(x) \}.
    \end{equation}
    We consider the following sets:
    \begin{align}
        \alpha^i_{\tin}(H) &= \{ h \in H^i \mid \exists \, x \in \Theta(H) \ \text{s.t.} \ u^i(h, x^{-i}) \geq u^i(x) \} \\
        \alpha^i_{\tout}(H) &= \{ h \in S^i \setminus H^i \mid \exists \, x \in \Theta(H) \ \text{s.t.} \ u^i(h, x^{-i}) \geq u^i(x) \}.
    \end{align}
    One can see that $\alpha^i(H) = \alpha^i_{\tin}(H) \cup \alpha^i_{\tout}(H)$,
    $\alpha^i_{\tin}(H) \subset H^i$, and $\alpha^i_{\tout}(H) \not\subset H^i$,
    which indicate 
    \begin{equation}
        \alpha^i(H) \subset H^i \iff \alpha^i_{\tout}(H) = \emptyset .
    \end{equation}
    Therefore it suffices to examine if $u^i(h, x^{-i}) \geq u^i(x)$ holds for $h \in S^i \setminus H^i$.

    Table \ref{tab:pm-P1} tells us that the reports of $w_2$ and $w_3$ do not change the outcome when $w_1$ reports $W1$ or $W2$.
    We can therefore assume that $w_2$ and $w_3$ report $W1$ without loss of generality, 
    and what we need to show is that $H = \{W1, W2\} \times \{W1\} \times \{W1\}$ is cuwbr.

    Let $x \in \Theta(H)$ be $x = \qty((p_1, 1 - p_1), (1), (1))$.
    Each element in $x$ denotes the probability distribution describing the mixed strategy of each player,
    and $p_1$ is the probability that $w_1$ reports $W1$.
    When all players use $x$, payoffs are $u^{w_1} = 25$ and $u^{w_2} = u^{w_3} = 1$.
    From table \ref{tab:pm-P1-W1}, following inequalities hold:
    \begin{equation}
        u^{w_1}(3, x^{-w_1}) = u^{w_1}(4, x^{-w_1}) = 5 < u^{w_1}(x), \quad
        u^{w_1}(5, x^{-w_1}) = u^{w_1}(6, x^{-w_1}) = 1 < u^{w_1}(x).
    \end{equation}
    Thus $\alpha^{w_1}(H) \subset H^{w_1}$ holds, completing the proof.
\end{proof}

\subsection{Nonclosure under weakly better replies of $H_1(\mu^W_2)$ and $H_2(\mu^W_2)$} \label{ssec:proof-p2}
First, we show that $H_1(\mu^W_2)$ of equation \eqref{eq:H1(muW2)} is not cuwbr.
\begin{proof}
    As in the preceding subsections, it suffices to examine if $u^i(h, x^{-i}) \geq u^i(x)$ holds for $h \in S^i \setminus H^i$.
    Table \ref{tab:pm-P2} tells us that the report of $w_3$ does not affect the outcome.
    We therefore assume that $w_3$ submits $W1$ without loss of generality,
    and what we need to show is that $H = \{W1, W2\} \times \{W6\} \times \{W1\}$ is not cuwbr.

    Let $x \in \Theta(H)$ be $x = \qty((p_1, 1 - p_1), (1), (1))$.
    Each element in $x$ denotes the probability distribution describing the mixed strategy of each player,
    and $p_1$ is the probability that $w_1$ reports $W1$.
    When all players use $x$, payoffs are $u^{w_1} = u^{w_2} = 25$, $u^{w_3} = 5$.
    From table \ref{tab:pm-P2-W1}, following expressions hold:
    \begin{equation}
        \begin{alignedat}{6}
            u^{w_1}(2, x^{-w_1}) = &\ u^{w_1}(3, x^{-w_1}) &&= 1 &&< u^{w_1}(x), & \quad
            u^{w_1}(4, x^{-w_1}) = &\ u^{w_1}(5, x^{-w_1}) &&= 5 &&< u^{w_1}(x) \\
            u^{w_2}(0, x^{-w_2}) = &\ u^{w_2}(1, x^{-w_2}) &&= 5 &&< u^{w_2}(x), &
            u^{w_2}(2, x^{-w_2}) = &\ 1 + 4 p_1 &&\leq 5 &&< u^{w_2}(x) \\
            &\ u^{w_2}(3, x^{-w_2}) &&= 1 &&< u^{w_2}(x), &
            u^{w_2}(4, x^{-w_2}) = &\ 5 + 20 p_1 &&\leq 25 &&= u^{w_2}(x).
        \end{alignedat}
    \end{equation}
    Therefore
    \begin{equation} \tag{\ref{eq:cuwbr-H1}}
        \begin{aligned}
            \alpha^{w_1}(H) &\subset H^{w_1} \\
            \alpha^{w_2}(H) &\subset H^{w_2} \cup \{W5\} \not\subset H^{w_2},
        \end{aligned}
    \end{equation}
    and $H$ is not cuwbr.
    \end{proof}

Next, we show that $H_2(\mu^W_2)$ of equation \eqref{eq:H2(muW2)} is not cuwbr either.
\begin{proof}
    In the same way as above, we examine if $u^i(h, x^{-i}) \geq u^i(x)$ holds for $h \in S^i \setminus H^i$.
    Table \ref{tab:pm-P2} tells us that the report of $w_3$ does not affect the outcome.
    We therefore assume that $w_3$ submits $W1$ without loss of generality,
    and what we need to show is that $H = \{W1\} \times \{W5, W6\} \times \{W1\}$ is not cuwbr.

    Let $x \in \Theta(H)$ be $x = \qty((1), (q_5, 1 - q_5), (1))$.
    Each element in $x$ denotes the probability distribution describing the mixed strategy of each player,
    and $q_5$ is the probability that $w_2$ reports $W5$.
    When all players use $x$, payoffs are $u^{w_1} = u^{w_2} = 25$, $u^{w_3} = 5$.
    From table \ref{tab:pm-P2}, following expressions hold:
    \begin{equation}
        \begin{alignedat}{4}
            u^{w_1}(1, x^{-w_1}) &= 25 - 20 q_5 && &&\leq 25 &&= u^{w_1}(x) \\
            u^{w_1}(2, x^{-w_1}) &= u^{w_1}(3, x^{-w_1}) && &&= 1 &&< u^{w_1}(x) \\
            u^{w_1}(4, x^{-w_1}) &= u^{w_1}(5, x^{-w_1}) && &&= 5 &&< u^{w_1}(x) \\
            u^{w_2}(0, x^{-w_2}) &= u^{w_2}(1, x^{-w_2}) &&= u^{w_2}(2, x^{-w_2}) &&= 5 &&< u^{w_2}(x) \\
            u^{w_2}(3, x^{-w_2}) &= 1 && && &&< u^{w_2}(x)
        \end{alignedat}
    \end{equation}
    Therefore
    \begin{equation} \tag{\ref{eq:cuwbr-H2}}
        \begin{aligned}
            \alpha^{w_1}(H) &\subset H^{w_1} \cup \{W2\} \not\subset H^{w_1} \\
            \alpha^{w_2}(H) &\subset H^{w_2},
        \end{aligned}
    \end{equation}
    and $H$ is not cuwbr.
\end{proof}

\subsection{Nonexistence of a partial preimage of $\mu^W_2$ closed under weakly better replies} \label{ssec:proof-p2-nocuwbr}
We have already shown that $H_1(\mu^W_2)$ and $H_2(\mu^W_2)$ are not cuwbr.
We prove that its other partial preimages are not cuwbr either to show that there is no partial preimage of $\mu^W_2$ that is cuwbr.
\begin{proof}
    In the $6^6$ possible report profiles for women, eighteen result in $\mu^W_2$,
    all of which are in $H_1(\mu^W_2) \cup H_2(\mu^W_2)$.
    Thus their subsets are the only candidates for a partial preimage that is cuwbr.
    Let us consider the partial preimage $H_3(\mu^W_2) = \{W1\} \times \{W6\} \times \{W1, \dots, W6\}$.
    Since the report of $w_3$ does not change the outcome, 
    we may assume that $w_3$ reports $W1$ without loss of generality.
    Then the partial preimage to examine is $H = \{W1\} \times \{W6\} \times \{W1\}$, which is a pure strategy profile.
    Since the cuwbr of $H$ is equivalent to that $H$ is a strict Nash equilibrium, and
    there is no strategy profile that is a strict Nash equilibrium in the current setting as we mentioned in section \ref{sec:sim},
    $H$ is not cuwbr.
    The same argument holds for all other partial preimages of $\mu^W_2$, and thus 
    there is no partial preimage of $\mu^W_2$ that is cuwbr.
\end{proof}

\section{Proof for section \ref{sec:disc}} \label{sec:proof-disc}
We show that a trajectory starting from $x' \in \Theta(H')$ ends up at some $x'' \in \Theta(H_1 \cup H_2)$ under payoff-positive selection dynamics.

\begin{proof}
First, observe that $H'$ is not a partial preimage of $\mu^W_2$, since its subset
\begin{equation}
    H_3 := \{W2\} \times \{W5\} \times \{W1, \dots, W6\}
\end{equation}
is a partial preimage of $\mu^M_2$.
Because $H' \setminus \qty(H_1 \cup H_2) = H_3$, trajectories starting in $\Theta(H')$ go into either $\Theta(H_1 \cup H_2)$ or $\Theta(H_3)$.
Now, we look at dynamics.
From the payoff matrices in Table \ref{tab:pm-P2}, one finds that the report of $w_3$ does not change the outcome.
Therefore we can assume $w_3$ always reports $W1$ without loss of generality.
When the state $x'$ is in $\Theta(H')$, $x^{w_1}_h = 0$ for $h \in \qty{W3, W4, W5, W6}$ and $x^{w_2}_h = 0$ for $h \in \qty{W1, W2, W3, W4}$.
Using $x^{w_1}_{W1} + x^{w_1}_{W2} = x^{w_2}_{W5} + x^{w_2}_{W6} = 1$, dynamics of $x'$ become two-dimensional dynamical systems:
\begin{equation}
    \begin{cases}
        \dot{x}^{w_1}_{W1} = g^{w_1}_{W1}(x) \, x^{w_1}_{W1} \\
        \dot{x}^{w_2}_{W6} = g^{w_2}_{W6}(x) \, x^{w_2}_{W6}.
    \end{cases}
\end{equation}
Recalling the definition of the payoff-positivity and looking at the payoff matrix in Table \ref{tab:pm-P2-W1}, we find
\begin{align}
    g^{w_1}_{W1}(x) \lessgtr 0 
    &\iff u^{w_1}(W1, x^{-w_1}) - u^{w_1}(x) \lessgtr 0 
    \iff 20 \qty(1 - x^{w_1}_{W1}) \qty(1 - x^{w_2}_{W6}) \lessgtr 0 \\
    g^{w_2}_{W6}(x) \lessgtr 0 
    &\iff u^{w_2}(W6, x^{-w_1}) - u^{w_2}(x) \lessgtr 0 
    \iff 20 \qty(1 - x^{w_1}_{W1}) \qty(1 - x^{w_2}_{W6}) \lessgtr 0.
\end{align}
Hence, $\dot{x}^{w_1}_{W1} > 0$ and $\dot{x}^{w_2}_{W6} > 0$ until either $x^{w_1}_{W1}$ or $x^{w_2}_{W6}$ gets to be one.
As $x \in \Theta(H_1)$ when $x^{w_2}_{W6} = 1$, and $x \in \Theta(H_2)$ when $x^{w_1}_{W1} = 1$,
we have established that trajectories starting in $\Theta(H')$ go into either $\Theta(H_1)$ or $\Theta(H_2)$ under payoff-positive selection dynamics.
\end{proof}

\end{document}